\documentclass[conference,compsoc]{IEEEtran}

\usepackage{amsthm}
\usepackage[hidelinks]{hyperref}
\usepackage{graphicx}
\usepackage{placeins}
\usepackage{caption}
\usepackage{diagbox}
\usepackage{hyperref}
\usepackage{setspace}
\usepackage{array}
\usepackage{amsfonts}           
\usepackage{mathrsfs}
\usepackage{comment}
\usepackage{xargs}
\usepackage{mathtools}
\usepackage{tabularx}
\usepackage{enumerate}
\usepackage{amsthm}
\usepackage{enumitem}
\setlist{nosep}
\usepackage{tikz}
\usepackage[htt]{hyphenat}
\usepackage{color}
\usepackage{lscape}
\usepackage{blindtext}
\usepackage[nocompress]{cite}
\usepackage{hyperref}
\hypersetup{
    breaklinks=true,
    bookmarksopen,
    bookmarksdepth=2,
    breaklinks=true
}

\usepackage{capt-of}
\usepackage{url}
\usepackage{xurl}
\usepackage[framemethod=TikZ]{mdframed}
\usepackage[font=small,labelfont=bf]{caption} 
\usepackage{rotating}
\usepackage[normalem]{ulem}

\usepackage{amsmath}
\usepackage{graphicx}
\usepackage{multirow}
\usepackage{subfigure}
\usepackage{filecontents}
\usepackage{footnote}
\usepackage{listings}
\usepackage{xspace}
\usepackage{booktabs}
\usepackage[ruled,vlined,linesnumbered,noend]{algorithm2e}
\SetKwComment{Comment}{$\triangleright$\ }{}
\SetKwInOut{Variables}{Variables}
\usepackage{soul}
\usepackage{algorithmicx}  
\usepackage{algpseudocode} 
\usepackage{pifont}
\usepackage[flushleft]{threeparttable}
\usepackage{xcolor}
\usepackage{makecell}
\usepackage{lipsum}
\usepackage{balance}
\usepackage{caption} 
\usepackage{enumitem}
\setlist{noitemsep}
\allowdisplaybreaks[4]
\definecolor{dkgreen}{rgb}{0,0.6,0}
\definecolor{gray}{rgb}{0.5,0.5,0.5}
\definecolor{mauve}{rgb}{0.58,0,0.82}
\definecolor{applegreen}{rgb}{0.55, 0.71, 0.0}
\definecolor{amber}{rgb}{1.0, 0.75, 0.0}
\definecolor{firebrick}{rgb}{0.7, 0.13, 0.13}
\definecolor{darkblue}{rgb}{0,0,0.55}

\def\eg{\emph{e.g.,}\xspace}

\def\ie{\emph{i.e.,}\xspace}
\def\etal{\emph{et al.}\xspace}

\def\Inf{\operatornamewithlimits{inf\vphantom{p}}}
\def\Sup{\operatornamewithlimits{sup\vphantom{p}}}

\usepackage[font=small,labelfont=bf]{caption}
\setlist{noitemsep}

\newcommand{\mathlower}[0]{\mathrm{lower}\xspace}
\newcommand{\mathupper}[0]{\mathrm{upper}\xspace}

\newcommand{\cpmax}[0]{C_p^{\max}}
\newcommand{\cpmin}[0]{C_p^{\min}}

\newtheorem{theorem}{Theorem}[section]

\newtheorem{definition}{Definition}[section]

\newtheorem{corollary}{Corollary}[section]
\newtheorem{remark}{Remark}[section]

\DeclareMathOperator{\randominit}{random\_init}
\DeclareMathOperator{\variance}{variance}
\DeclareMathOperator{\modify}{modify}

\begin{document}

\title{Bounded and Unbiased Composite Differential Privacy}

\author{\IEEEauthorblockN{Kai Zhang\IEEEauthorrefmark{1},
Yanjun Zhang\IEEEauthorrefmark{2}\IEEEauthorrefmark{3},
Ruoxi Sun\IEEEauthorrefmark{2}, 
Pei-Wei Tsai\IEEEauthorrefmark{1}, \\
Muneeb Ul Hassan\IEEEauthorrefmark{4}, 
Xin Yuan\IEEEauthorrefmark{2}, 
Minhui Xue\IEEEauthorrefmark{2}, and
Jinjun Chen\IEEEauthorrefmark{1}}
\IEEEauthorblockA{\IEEEauthorrefmark{1}Swinburne University of Technology, Australia} 
\IEEEauthorblockA{\IEEEauthorrefmark{2}CSIRO's Data61, Australia}
\IEEEauthorblockA{\IEEEauthorrefmark{3}University of Technology Sydney, Australia}
\IEEEauthorblockA{\IEEEauthorrefmark{4}Deakin University, Australia}}

\maketitle

\begin{abstract}
The objective of differential privacy (DP) is to protect privacy by producing an output distribution that is indistinguishable between any two neighboring databases. 
However, traditional differentially private mechanisms tend to produce unbounded outputs in order to achieve maximum disturbance range, which is not always in line with real-world applications. Existing solutions attempt to address this issue by employing post-processing or truncation techniques to restrict the output results, but at the cost of introducing bias issues. In this paper, we propose a novel differentially private mechanism which uses a composite probability density function to generate bounded and unbiased outputs for any numerical input data. The composition consists of an activation function and a base function, providing users with the flexibility to define the functions according to the DP constraints. We also develop an optimization algorithm that enables the iterative search for the optimal hyper-parameter setting without the need for repeated experiments, which prevents additional privacy overhead. Furthermore, we evaluate the utility of the proposed mechanism by assessing the variance of the composite probability density function and introducing two alternative metrics that are simpler to compute than variance estimation. Our extensive evaluation on three benchmark datasets demonstrates consistent and significant improvement over the traditional Laplace and Gaussian mechanisms. The proposed bounded and unbiased composite differentially private mechanism will underpin the broader DP arsenal and foster future privacy-preserving studies.
\end{abstract}

\section{Introduction}
Differential privacy (DP) has become one of the most widely adopted privacy-preserving techniques and is now an essential technology in real-world applications. It originates from the concept of semantic security in cryptography, which aims to ensure that any attacker cannot distinguish individual information with respect to a specific dataset. DP mechanisms exhibit several characteristics. First, DP has a rigorous mathematical definition and a formal measure of privacy loss~\cite{dwork2011firm}. Second, most DP mechanisms can generate mathematically unbiased noise to preserve privacy while still maintaining the utility of the perturbed datasets for statistical analysis. Furthermore, utilizing DP mechanisms does not typically require significant additional computational resources. The merits of DP have led to its widespread adoption in various privacy-sensitive tasks, including healthcare~\cite{wang2022privacy, tang2019secure,rajput2021privacy}, finance~\cite{ma2019privacy}, government~\cite{piao2019privacy,sarathy2023don}, social media~\cite{jiang2021applications,wang2016real,qin2017generating}, and Internet of Things~\cite{Yang2022,Sun2020,Zheng2020}, and it has been increasingly incorporated by top tech giants, such as Microsoft~\cite{2}, Google~\cite{3}, and Apple~\cite{4,5}. 

The implementation of DP mechanisms usually involves adding noise or perturbation to raw data before releasing it, resulting in an output distribution that is indistinguishable from that of neighboring databases. Specifically, the randomness of the perturbation function used in DP mechanisms causes the overlapping interval of density functions between neighboring databases to become indistinguishable. This prevents attackers from inferring specific individual data by comparing output results, thus guarding against background knowledge attacks.

\begin{figure}[t]
\centering
\includegraphics[width=\linewidth]{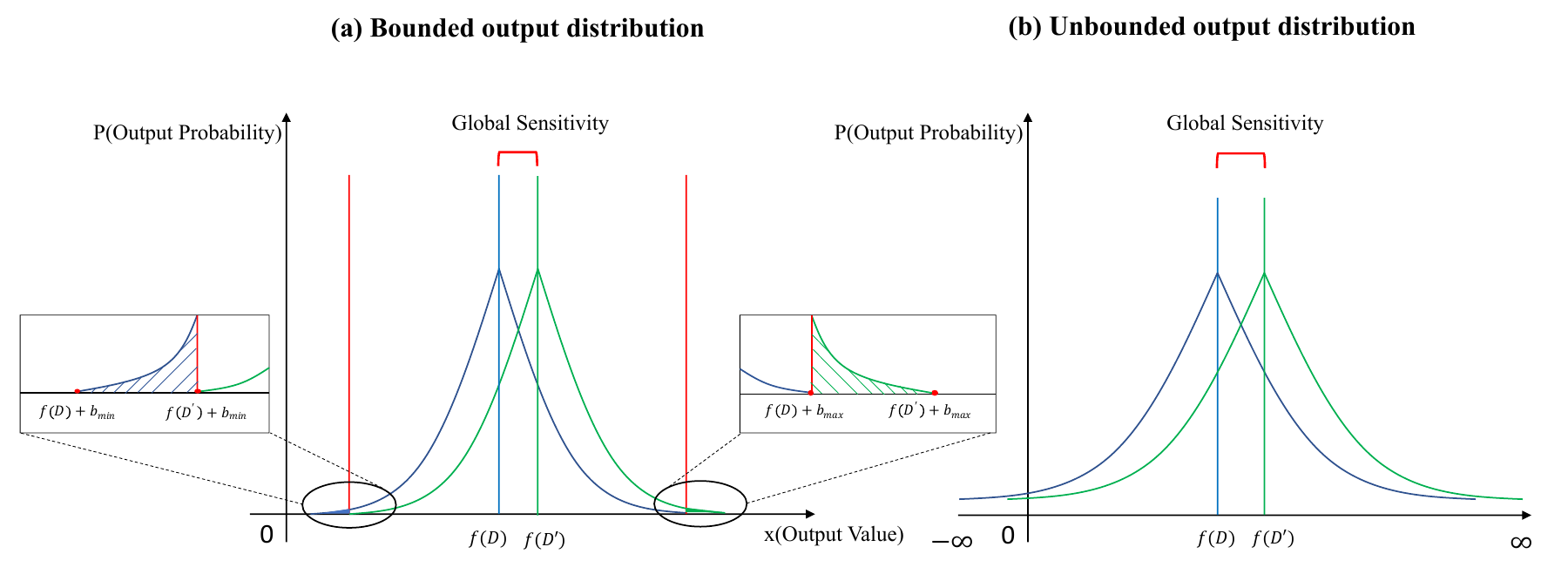}
\caption{Bounded and unbounded output distribution for neighboring databases $D$ and $D^\prime$.}\label{fig_output_distribution}
\end{figure}

To maintain indistinguishable characteristics, the output ranges after applying DP on querying neighboring databases should be the same. As demonstrated in Figure~\ref{fig_output_distribution}, we assume that the blue and green curves denote the output distributions of $\mathcal{M}(f(D))$ and $\mathcal{M}(f(D^\prime))$, respectively, where $D$ and $D^\prime$ are neighboring databases, $f(\cdot)$ is the query function, and $M(\cdot)$ refers to the differentially private mechanism applied. If DP applied to neighboring databases results in different output ranges (\eg $[f(D)+b_{\mathrm{min}},f(D)+b_{\mathrm{max}}]$ and $[f(D^\prime)+b_{\mathrm{min}},f(D^\prime)+b_{\mathrm{max}}]$ in Figure~\ref{fig_output_distribution}(a), where $b_{\mathrm{min}}$ and $b_{\mathrm{max}}$ refer to the range of perturbation), the indistinguishability requirement of DP is violated and the data privacy is compromised, since once the output result locates in $[f(D)+b_{\mathrm{min}},f(D^\prime)+b_{\mathrm{min}}]$, the attacker can firmly assert that the input database is $D$. Therefore, traditional DP mechanisms, such as the Laplace mechanism and Gaussian mechanism, have unbounded output range $(\mathrm{Lap}, \mathrm{Gaus} \in(-\infty,+\infty))$ for numerical input data (as shown in Figure~\ref{fig_output_distribution}(b)). However, an unbounded output distribution in DP mechanisms can pose an issue: \textbf{The unbounded output space contradicts the practical constraints imposed by real-world input data. } 

Real-world datasets commonly exhibit restrictions within their range. For instance, the average height of a group of children cannot be negative, and the median age of a patient group cannot exceed 150. To accommodate these constraints, post-processing methods or truncated DP mechanisms are often used to constrain the output results within a specific range. Specifically, the post-processing technique is applied to the outputs of perturbed data, which is then projected onto the feasible range~\cite{10,11,12,13,14}. In contrast, the truncated approach directly modifies the range of the probability density function to adhere to the valid region of the query~\cite{37,36,38}. However, using the existing post-processing and truncated approaches can result in bias issues~\cite{6}. In addition, when applying a relatively smaller privacy budget, the noise generated by traditional DP mechanisms becomes uncontrollable, resulting in a significant loss of statistical utility. Therefore, post-processing or truncated approaches are not ideal solutions for real-world DP applications, and developing a bounded and unbiased perturbation mechanism remains a non-trivial challenge.

In this paper, we introduce a novel composite differentially private mechanism that is both bounded and unbiased, and can be used to conduct controllable perturbation for any numerical input data. Unlike traditional DP mechanisms, our approach does not require post-processing or truncating steps to restrict the output space, thereby avoiding the bias problems that can arise from these approaches. Concretely, the proposed mechanism is comprised of an activation function and a base function, which provides users with the flexibility to define the perturbation function. Additionally, we build an optimization algorithm to iteratively modify the hyper-parameter setup to enhance data utility. The iterative function can evaluate the deviation degree of the dataset through theoretical computation, which mitigates privacy overhead that arises from repeated experiments. 

To comprehensively evaluate the performance of the proposed mechanism, we conduct both theoretical analysis and empirical experiments. Theoretical analysis of the proposed mechanism involves an assessment of the variance of the composite probability density function. We also introduce two alternative metrics that are easier to compute than variance estimation to further simplify the evaluation of the mechanism's utility. In the empirical evaluation, we conduct extensive experiments on three benchmark datasets. Our results demonstrate consistent and significant improvement of the proposed mechanism over the traditional Laplace and Gaussian mechanisms. 

In summary, the main contributions of this paper are as follows:
\begin{itemize}[leftmargin=*]
\item We propose a novel differentially private mechanism, providing bounded and unbiased outputs for any numerical input data. It is achieved by composing an activation function which adjusts the output expectation to ensure the unbiased characteristic, and a base function which constructs the constant probability distribution for all outputs to guarantee the bounded feature and satisfy the definition of DP. To the best of our knowledge, the proposed solution is the \textit{first} practical differentially private mechanism with bounded and unbiased outputs for real use cases.
\item The proposed mechanism offers users the flexibility to define the perturbation function based on real-world scenarios. In design, we present six perturbation functions that combine three different activation functions and two base functions. In evaluation, we introduce two alternative metrics that are simpler to compute than variance estimation for evaluating the utility of the mechanism. We thoroughly verify the privacy-utility performance as well as its correctness and effectiveness both theoretically and empirically. 
\item In implementation, we develop an optimization algorithm that iteratively modifies the hyper-parameter setup to enhance data utility. It employs theoretical computation to evaluate the performance in iteration, eliminating the need for repeated experiments and preventing the privacy overhead arising from them.
\item We release the source code and the artifact at \url{https://github.com/CompositeDP/CompositeDP}, which creates a new tool for the broader DP arsenal to enable future privacy-preserving studies.
\end{itemize}

\section{Preliminaries}
Differential privacy (DP) is derived from the concept of semantic security in cryptography, in which the attacker cannot distinguish whether an individual exists in the specific dataset~\cite{8,dwork2011firm}. 

\begin{definition}[\textbf{$\epsilon$-DP}]
A randomized mechanism $\mathcal{M}: \mathcal{\tilde D} \rightarrow \mathcal{R}$ with domain $\mathcal{\tilde D}$ and range $\mathcal{R}$ satisfies $\epsilon$-DP if for any two neighboring datasets $D$, $D^\prime \in \mathcal{\tilde D}$ where $\|D-D^\prime\|_1=1$, and for any subset of outputs $S \subseteq \mathcal{R}$ it holds that
\begin{equation}\label{equ:epdeldp}
\Pr[\mathcal{M}(D) \in S] \leq e^\epsilon \Pr[\mathcal{M}(D^\prime) \in S].
\end{equation}
\end{definition}

\begin{definition}[\textbf{Laplace Mechanism}]
Given any function $f: \mathcal{\tilde D} \rightarrow \mathcal{R}$, the Laplace mechanism is defined as
\begin{equation}
\mathcal{M}_L(D, f(\cdot), \epsilon) = f(D) + Lap(\Delta f/\epsilon), 
\end{equation}
where $Lap(\Delta f/\epsilon)$ denotes the independent and identically distributed (i.i.d) random variables drawn from the Laplace distribution, and $\Delta f$ denotes the $\ell_1$-sensitivity of $f(\cdot)$
\begin{equation}
\Delta f = \max\limits_{\|D-D^\prime\|_1=1} \|f(D)-f(D^\prime)\|_1.
\end{equation}
\end{definition}

When applying the Gaussian noise mechanism~\cite{9,abadi2016deep}, the definition permits the possibility that plain $(\epsilon,0)$-DP can be broken with a preferably subpolynomially small probability $\delta$.  

\begin{definition}[\textbf{Gaussian Mechanism}]
Given any function $f: \mathcal{\tilde D} \rightarrow \mathcal{R}$, the Gaussian mechanism is defined as
\begin{equation}
\mathcal{M}_G(D, f(\cdot), \epsilon, \delta) = f(D) + \mathcal{N}(0,\sigma^2), 
\end{equation}
where
\begin{equation}
\sigma > \frac{\sqrt{2\ln{\frac{1.25}{\delta}}}\Delta_2 f}{\epsilon}. 
\end{equation}
Here, $\Delta_2 f$ is the $\ell_2$-sensitivity and given by 
\begin{equation}
\Delta_2 f = \max\limits_{\|D-D^\prime\|_1=1} \|f(D)-f(D^\prime)\|_2.
\end{equation}
\end{definition}

\begin{definition}[\textbf{Unbiased Perturbation}]\label{def:unbias}
Given any perturbation mechanism $\mathcal{M}$, let $P$ denote the probability density function of $\mathcal{M}$, and $f(D)$ denote the input to $\mathcal{M}$. To achieve an unbiased perturbation at any arbitrary point $x$ within the output space of $\mathcal{M}$, \ie $x \in \mathcal{R}$, the following condition must hold
\begin{equation}
\label{equ:unbiasunbound}
E(x) =  \int_{-\infty}^\infty P(x)x dx = f(D),
\end{equation}
with an unbounded output space, and 
\begin{equation}
\label{equ:unbias}
E(x) =  \int_{x_{\mathrm{min}}}^{x_{\mathrm{max}}} P(x)x dx = f(D), 
\end{equation}
with a bounded output $x\in[x_{\mathrm{min}},x_{\mathrm{max}}]$. 
\end{definition}

\section{The Composite Differentially Private Mechanism}
In this section, we begin by formalizing the composite differentially private mechanism and prove its unbiasedness and privacy guarantee, followed by the introduction of six examples of perturbation functions, with the components of activation functions and base functions. We finally verify the differential privacy guarantee of the proposed perturbation functions. 

\subsection{Formalizing the Composite Differentially Private Mechanism}
Traditional approaches for DP, such as the Laplace and the Gaussian mechanisms, are two special cases with unbounded outputs. However, in real-world scenarios, the Laplace and the Gaussian mechanisms inevitably introduce post-processing or truncating problems. Specifically, when users force the perturbed output value to fit within a reasonable range that depends on its data attribute, unpredictable bias issues arise which distort the statistical analysis results.

In this work, we propose a novel DP mechanism that naturally incorporates the bounded output space. Considering that the probability density function (the phrase ``probability density function'' is referred to as ``perturbation function'' below) of the perturbation mechanism is a key component of DP, an intuitive solution is to bind the independent variable of the perturbation function in a specific range.
As such, the fundamental difference between the proposed composite DP mechanism and traditional mechanisms lies in whether the independent variable of the perturbation function is bounded. Additionally, the proposed perturbation function should also ensure statistically unbiased output for the perturbation, \ie satisfying~Equation~\eqref{equ:unbias}. 

To construct a perturbation function that satisfies boundedness and unbiasedness, we resort to composite and piece-wise functions due to their flexibility. The combination of these functions can provide a variety of perturbation functions that meet the requirements of various use cases. We design a perturbation function that consists of an activation function $H(x)$ and a base function $G(x)$. In the following, we first introduce the formal definition of the Composite DP Mechanism. 

\begin{definition}[\textbf{Composite Differentially Private Mechanism}]\label{def-mechanism}
A randomized composite mechanism $\mathcal{M}_c: \mathcal{\tilde D} \rightarrow \mathcal{N}$ with domain $\mathcal{\tilde D}$  and range $\mathcal{N} = \{c \in \mathbb{R}|l\leq c \leq u\}$, for any database $D$, query $f(\cdot)$, is defined as
\begin{equation}
\mathcal{M}_c(D, f(\cdot), \epsilon, l, u) = {\rm{Comp}}(f(D), \Delta f, \epsilon, l, u), 
\end{equation}
where $ {\rm{Comp}}(f(D), \Delta f, \epsilon, l, u)$ denotes the independent and identically distributed (i.i.d) random variables drawn from the  distribution, and $\Delta f$ is the $\ell_1$-sensitivity of $f(\cdot)$.
\end{definition}

The composite differentially private mechanism involves three main steps. Firstly, a mapping procedure is employed to rescale range $\mathcal{N} = \{c \in \mathbb{R}|l\leq c \leq u\}$ to domain~ $\mathcal{L} = \{x \in \mathbb{R} |-L\leq x \leq L\}$. This process is similar to the clipping technique in previous works~\cite{abadi2016deep,wang2019collecting,5,23}. Secondly, the output of the mapping is perturbed by a perturbation function defined over $\mathcal{L}$. Thirdly, an inverse mapping function is used to map the perturbed results from $\mathcal{L}$ back into $\mathcal{N}$. In the following, we provide a detailed explanation of them, respectively. 

We first introduce some notations before stating the mapping process. 
Let $[\cpmin, \cpmax]$ be a valid input range within $\mathcal{L}$. The values of $\cpmin$ and $\cpmax$ are determined by different shapes of the perturbation function, the detailed derivations of which are provided in Appendix~\ref{appendix-A}.   
The $\mathlower$ and $\mathupper$ bounds of the output space of $f$, denoted by $l$ and $u$, can be chosen manually according to the datasets' nature. For instance, when querying a dataset for the maximum age of teenagers, $l$ can be set as 10, or alternatively, users can set $u$ as 30. The choice of these values, influenced by the nature of the dataset, is independent of the input data.

\begin{remark}
    The mapping from real space $\mathcal{N}$ to domain~$\mathcal{L}$ is characterized as a linear transformation. Explicitly, with the $\mathlower$ and $\mathupper$ bounds $l$ and $u$ of $\mathcal{N}$, and the $\mathlower$ and $\mathupper$ bounds $-L$ and $L$ of $\mathcal{L}$, through mapping the sensitivity $\Delta f$ in range $\mathcal{N}$ to the corresponding sensitivity in domain $\mathcal{L}$, the association between the $\mathlower$ and $\mathupper$ can be established by letting $u = l + \frac{2L}{\cpmax - \cpmin} \Delta f$. Either $l$ or $u$ can be utilized to conduct the mapping process. The mapping function is defined as follows.
\end{remark}

\begin{definition}[\textbf{Mapping Function}]
Given the $\ell_1$-sensitivity $\Delta f$,  $\cpmin$, $\cpmax$, $l$ and $u$, for any input $c\in \mathcal{N}$, the mapping function $\gamma: \mathcal{N} \rightarrow \mathcal{L} = \{x \in \mathbb{R} |-L\leq x \leq L\}$ is defined as
\begin{equation}\label{map_fun_l}
\begin{split}
    \gamma(c) = (c - l)C - L,
\end{split}
\end{equation}
where $C = \frac{\cpmax - \cpmin}{\Delta f}$.
By substituting $u = l+ \frac{2L}{\cpmax-\cpmin} \Delta f$ into \eqref{map_fun_l}, we further have
\begin{equation}
    \gamma(c) = (c - u)C+L.
\end{equation}
\end{definition}

\begin{remark} 
The mapping function facilitates the enhancement of data consistency, enabling users to have better control over perturbation. Moreover, a unified domain simplifies the process of setting constraints based on the specific scenarios. All relevant parameters are solely dependent on the query $f$, $\epsilon$ and the determination of the perturbation function, thereby ensuring the compliance of privacy.
\end{remark}

In domain $\mathcal{L}$, the output of the mapping from the real space $\mathcal{N}$ is perturbed to protect data privacy. The perturbation function is defined as follows.
\begin{definition}[\textbf{Perturbation Function}]\label{def-perturbation}
For any two neighboring datasets $D, D^\prime \in \mathcal{\tilde D}$ and query $f$, and an activation function $H(x)$ and a base function $G(x)$ with a bounded independent variable $x\in \mathcal{L}$, we call the perturbation function of $\mathcal{M}_c,$ $P(x) \coloneqq H(x)+G(x)$, if the following are true:    
\begin{itemize}
    \item To satisfy unbiased output, the expectation output of perturbation function $P$ should be equal to~$C_p$, i.e., $E(x) = C_p$, where $C_p$ denotes the mapped input data corresponding to $f(D)$. 

    \item 
    To satisfy differential privacy guarantee, $\forall x, x' \in \mathcal{L}$, the corresponding perturbation functions $P(x)$ and $P'(x')$ should satisfy $\frac{\Sup_{x\in \mathcal{L}}~P(x)}{\Inf_{x^\prime\in \mathcal{L}}~P'(x')} \leq e^{\epsilon}$, where  $P$ and $P'$ denote the perturbation function associated with $C_p$ and $C'_p$, respectively, and $C_p, C'_p \in \mathcal{L}$ denote the mapped input data associated with $f(D)$ and $f(D^\prime)$, respectively.  
\end{itemize}
\end{definition}

\begin{remark}
    It is noted that Definition~\ref{def-perturbation} constructs the perturbation function for any input data to be within domain~$\mathcal{L}$. Without loss of generality, this definition focuses on the mechanism's output in a one-dimensional space defined over $\mathbb{R}$. 
Nonetheless, extending this definition to higher dimensional spaces is straightforward and can be readily achieved. 
\end{remark}

After perturbation, an inverse mapping function is used to map the perturbed results from $\mathcal{L}$ back into $\mathcal{N}$. The definition of the inverse mapping function is as follows. 
\begin{definition}[\textbf{Inverse Mapping Function}]
\label{def:inverse}
$\forall x \in \mathcal{L}$,
an inverse mapping function $\gamma^{-1} : \mathcal{L} \rightarrow \mathcal{N}$ is defined as
\begin{equation}\label{map_inverse_fun_l}
    \gamma^{-1}(x) = \frac{x+L}{C}+l,
\end{equation}
which can also be written as
\begin{equation}
    \gamma^{-1}(x) = \frac{x-L}{C}+u.
\end{equation}
\end{definition}

The proposed composite differentially private mechanism could provide unbiased perturbation while ensuring DP, as elaborated in the following two theorems.
\begin{theorem}[\textbf{Unbiased Perturbation}]\label{unbiased_theorem}
The proposed Composite Differentially Private Mechanism is an unbiased perturbation.
\end{theorem}
\begin{proof}
According to Definition~\ref{def-perturbation}, let $P$ denote the perturbation function corresponding to $C_p \in \mathcal{L}$. Given a dataset $D$, the raw query result $f(D)$, an activation function $H(x)$ and a base function $G(x)$. 
According to Definition~\ref{def:inverse}, for $\forall c \in \mathcal{N}$, we have $c = \gamma^{-1}(x)$, where $x \in \mathcal{L}$. Therefore,
\begin{subequations}\label{eq_unbiased}
\begin{align}
    E(c) & = E \left(\gamma^{-1}(x) \right) \label{eq_unbiased a} \\
        & = \gamma^{-1}(E(x))  \label{eq_unbiased b}\\
        &=  \gamma^{-1}\left(\int_{-L}^{L} P(x)x dx\right) \label{eq_unbiased c}\\
        & =  \gamma^{-1}\left(\int_{-L}^{L} H(x)x dx + \int_{-L}^{L} G(x)x dx\right) \label{eq_unbiased d}\\
        & =  \gamma^{-1}(C_p) \label{eq_unbiased e}\\
        & =  f(D), \label{eq_unbiased f}
\end{align} 
\end{subequations}
where 
\eqref{eq_unbiased b} is due to the linearity of the inverse mapping function $\gamma^{-1}$, and \eqref{eq_unbiased f} is obtained by inverse mapping the perturbed output $C_p$ into $f(D) \in \mathcal{N}$, \ie $ \gamma^{-1}(C_p) = f(D)$.
\end{proof}

\begin{theorem}[\textbf{$\epsilon$-DP}]\label{DP_theorem}
The proposed Composite DP mechanism preserves~{$\epsilon$-DP}.
\end{theorem}

\begin{proof}
Given any two neighboring databases $D, D'$ and the query~$f$, 
for any subset of outputs $S \subseteq \mathcal{N}$,  
we have
\begin{subequations}\label{eq-dp proof}
\begin{align}
    & \frac{\Pr(\mathcal{M}_c(D) \in S)}{\Pr( \mathcal{M}_c(D^\prime) \in S)} \\
    & = \frac{P(v|\gamma^{-1}(v) \in S)}{P'(v'|\gamma^{-1}(v') \in S)} \label{eq-dp proof b}\\
    & =  \frac{H(v|\gamma^{-1}(v) \in S)+G(v|\gamma^{-1}(v) \in S)}{H'(v'|\gamma^{-1}(v') \in S)+G(v'|\gamma^{-1}(v') \in S)} \\
    & =  \frac{\int_{u \in S}\left[ H(v|\gamma^{-1}(v)=u )+G(v|\gamma^{-1}(v) =u )\right]du}{\int_{u \in S}\left[ H'(v'|\gamma^{-1}(v') =u )+G(v'|\gamma^{-1}(v') =u )\right]du } \\
    & \leq 
     \frac{\int_{u \in S} \displaystyle\Sup_{x\in \mathcal{L}}~P(x) du}
    {\int_{u \in S} \displaystyle\Inf_{x^\prime\in \mathcal{L}}~P'(x') du}  = \frac{\displaystyle\Sup_{x\in \mathcal{L}}~P(x)\cdot\int_{u \in S}1 du}
    { \displaystyle\Inf_{x^\prime\in \mathcal{L}}~P'(x') \cdot\int_{u \in S}1 du}
    \\
    & \leq e^{\epsilon},
\end{align}
\end{subequations}
where $v$ and $v'$ are the perturbed outcomes of $P$ and $P'$, respectively, and \eqref{eq-dp proof b} is obtained by letting $\gamma^{-1}(v) \in S$ and $\gamma^{-1}(v')\in S$. Here, for any point $u \in S$, $\gamma^{-1}(v) =u $ and $\gamma^{-1}(v') = u $ indicate inverse mapping of $v$ and $v'$ in domain $\mathcal{L}$ back into $u \in S$ in real space $\mathcal{N}$.

\end{proof}

\subsection{Perturbation Functions in Construction}
\label{sec:perfunc}

As we recall Definition~\ref{def-perturbation}, the perturbation function is constructed by the combination of the activation function and base function. 
The activation function is designed to adjust the expectation value of the perturbation function so as to satisfy the unbiasedness, while the base function remains constant for all input data, creating indistinguishable areas with a bounded range. 

\begin{figure*}[t]
\centering
\begin{minipage}[t]{0.55\linewidth}
\includegraphics[width=\linewidth]{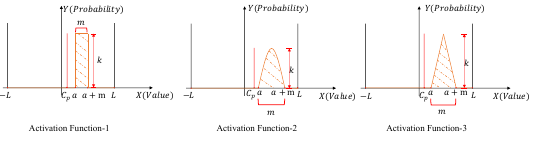}
\caption{Examples of proposed activation functions.}
\label{fig_act_fun}
\end{minipage}\quad\quad
\begin{minipage}[t]{0.4\linewidth}
\includegraphics[width=\linewidth]{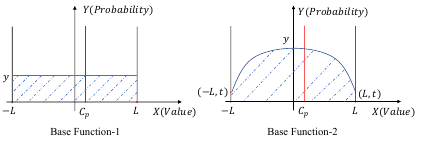}
\caption{Examples of proposed base functions.}
\label{fig_base_fun}
\end{minipage}
\end{figure*}

\noindent \textbf{Activation Functions.~}\label{subsection_act}
To satisfy unbiased output, it requires that the expectation of a perturbation function equals to the input data of the mechanism (Definition~\ref{def-perturbation}). Since the valid range of independent variables in traditional DP mechanisms (\eg Laplace and Gaussian) is unbounded, the unbiased criteria could be satisfied as the perturbation functions are always symmetric with respect to $C_p$. However, with a bounded output space (being restricted to domain $\mathcal{L}$), it is infeasible to construct the perturbation function in the same way since $C_p$ is not necessary to be at the center of domain $\mathcal{L}$. Therefore, we use an activation function to adjust the expectation of the output. Concretely, the activation function is a piece-wise function composed of several parts. In this section, we describe three general examples of proposed activation functions (Figure~\ref{fig_act_fun}).

\noindent \textbf{Activation Function-1 (\texttt{A1}).} 
Given $C_p, a \in \mathcal{L}$  determined by $C_p$, $k\in [0,1)$, $m \geq 0, L > 0$, $\forall x \in \mathcal{L}$, activation function $H(x)$ is as follows
\begin{equation}
    H(x) = 
\begin{cases} 
    0 \quad & \text{\textnormal{if }} x\in[-L,a), \\
    k \quad & \text{\textnormal{if }} x\in[a,a+m), \\
    0 \quad & \text{\textnormal{if }} x\in[a+m,L],
\end{cases}
\end{equation}
with the integral of \texttt{A1} for $x \in \mathcal{L}$ being
\begin{equation}
    S_1 = \int_{-L}^{L} H(x) dx = \int_{a}^{a+m} k dx = km.
\end{equation}

\noindent \textbf{Activation Function-2 (\texttt{A2}).}
Given $C_p, a \in \mathcal{L}$ determined by $C_p$, $k\in [0,1)$, $m \geq 0, L > 0$, $\forall x \in \mathcal{L}$, activation function $H(x)$ is as follows
\begin{equation}
    H(x) = 
\begin{cases} 
    0 \quad & \text{\textnormal{if }} x\in[-L,a), \\
    k\mathrm{sin}\frac{\pi}{m}(x-a) \quad & \text{\textnormal{if }} x\in[a,a+m), \\
    0 \quad & \text{\textnormal{if }} x\in[a+m,L],
\end{cases}
\end{equation}
with the integral of \texttt{A2} for $x \in \mathcal{L}$ being
\begin{equation}
    S_1 = \int_{-L}^{L} H(x) dx = \int_{a}^{a+m} k\mathrm{sin}\frac{\pi}{m}(x-a) dx = \frac{2km}{\pi}.
\end{equation}

\noindent \textbf{Activation Function-3 (\texttt{A3}).} 
Given $C_p, a \in \mathcal{L}$ determined by $C_p$, $k\in [0,1)$, $m \geq 0, L > 0$, $\forall x \in \mathcal{L}$, activation function $H(x)$ is as follows
\begin{equation}
    H(x) = 
\begin{cases} 
    0 \quad & \text{\textnormal{if }} x  \in[-L,a), \\
    \frac{2k}{m}x-\frac{2ak}{m} \quad & \text{\textnormal{if }}  x \in[a,a+\frac{m}{2}), \\
    -\frac{2k}{m}x+\frac{2(a+m)k}{m} \quad & \text{\textnormal{if }} x \in[a+\frac{m}{2},a+m), \\
    0 \quad & \text{\textnormal{if }} x  \in[a+m,L],
\end{cases}
\end{equation}
with the integral of \texttt{A3} for $x \in \mathcal{L}$ being
\begin{subequations}
\begin{align}
    S_1 & = \int_{-L}^{L} H(x) dx \\ 
    & = \int_{a}^{a+\frac{m}{2}} \frac{2k}{m}x - \frac{2ak}{m} dx \! +\! \int_{a+\frac{m}{2}}^{a+m} \!-\frac{2k}{m}x + \frac{2(a+m)k}{m} dx \\ 
    & = \frac{km}{2}.
\end{align}
\end{subequations}

As shown in the construction, the activation function only has probabilistic outputs within the range of $[a, a+m]$. As such, we can adjust the expected output of the final perturbation by tuning the parameter $a$ for any given $C_p \in \mathcal{L}$, 
thereby ensuring the unbiasedness of the proposed mechanism.

\noindent \textbf{Base Functions.~}\label{subsection_base}
The base function remains constant for all input data, maintaining an unaltered form, and it yields a positive output for all independent variables. 
As required in Definition~\ref{def-perturbation}, the ratio of the maximum and minimum output probability of the perturbation functions with respect to any $x$ and $x^\prime$ should be bounded by $e^{\epsilon}$. 
However, traditional Laplace and Gaussian functions do not satisfy this requirement as the infimum of their perturbation functions is zero. 
As such, in the following constructions, we use the base functions to bound the minimum output probability to enable the privacy guarantee, \ie $\Inf_{x \in \mathcal{L}} G(x) > 0$.  

Figure~\ref{fig_base_fun} shows two examples of base functions, and their definitions are provided as follows. 

\noindent \textbf{Base Function-1 (\texttt{B1}).}
Given $L>0, y>0$, $\forall x \in \mathcal{L}$, base function $G(x)$ is as follows:
\begin{equation}
    G(x) = y, \quad x\in \mathcal{L}, 
\end{equation}
with the integral of \texttt{B1} for $x \in \mathcal{L}$ being
\begin{equation}
    S_2 = \int_{-L}^{L} G(x) dx = 2yL.
\end{equation}

\noindent \textbf{Base Function-2 (\texttt{B2}).}
Given $L>0, y>0, t>0$, $\forall x \in \mathcal{L}$, base function $G(x)$ is as follows
\begin{equation}
    G(x) = -\frac{y-t}{L^4} x^4 + y, \quad x\in \mathcal{L}, 
\end{equation}
with the integral of \texttt{B2} for $x \in \mathcal{L}$ being
\begin{equation}
    S_2 = \int_{-L}^{L} G(x) dx = \frac{2L(t-y)}{5}+2Ly.
\end{equation}

\noindent \textbf{Perturbation Functions.~}\label{subsection_perturbation}
With these activation and base functions, we can further
construct six distinct perturbation functions following Definition~\ref{def-perturbation}. Each perturbation function is composed of an activation function $H(x)$, and a base function $G(x)$, as demonstrated in Figure~\ref{fig_perturbation_function}. Due to the inherent characteristics of the perturbation function (with its maximum value constrained not to exceed $1$, and the integral value over the valid range of the independent variable being $1$), each perturbation function is subject to four constraints (as detailed in Appendix~\ref{appendix-A}) while simultaneously satisfying the conditions of unbiasedness and DP guarantee, as detailed in Table~\ref{table1}. Only perturbation functions that meet these constraints are considered viable.

\begin{figure*}[t]
\centering
\captionof{table}{Six proposed perturbation functions, each comprises an activation function $H(x)$ and a base function $G(x)$. The perturbation function has specific characteristics, including its maximum value bounded by one and an integral value over the valid range of the independent variable that equals one. Moreover, all the perturbation functions are constructed to satisfy the DP guarantee. Finally, the variances of these proposed perturbation functions are calculated and listed in the last row. }\label{table1}
\includegraphics[width=\linewidth]{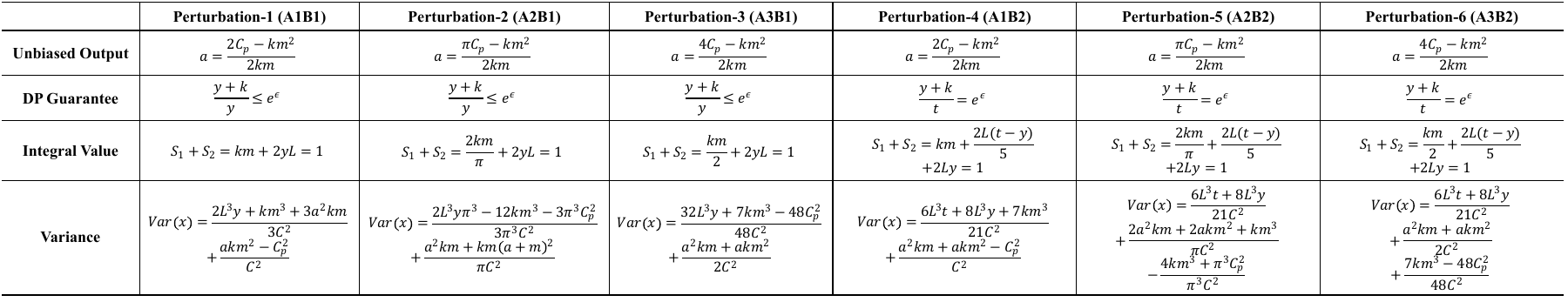}
\end{figure*}

\begin{figure*}[t]
\centering
\includegraphics[width=\textwidth]{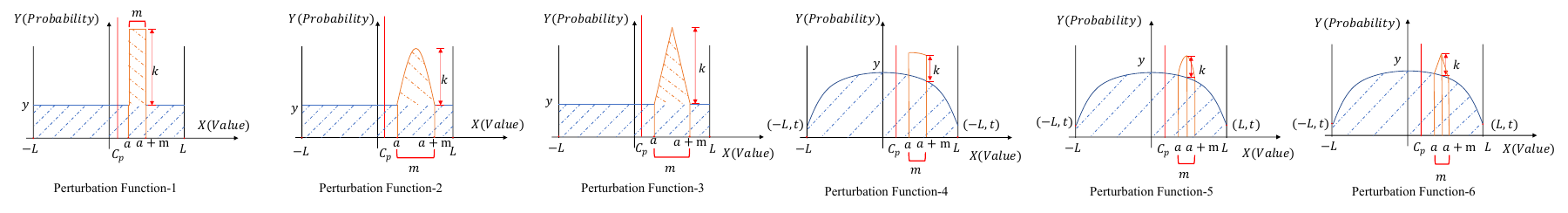}
\caption{Six instances of perturbation functions.}\label{fig_perturbation_function}
\end{figure*}

\noindent \textbf{Perturbation Function-1 (\texttt{A1B1}).}
Given $C_p, a \in \mathcal{L}$ determined by $C_p$, $k\in [0,1)$, $m \geq 0, L>0, y>0$, $\forall x \in \mathcal{L}$, perturbation function $P(x)$ is as follows
\begin{equation}\footnotesize
    P(x) = 
\begin{cases} 
    y \quad & \text{\textnormal{if }} x\in[-L,a), \\
    y+k \quad & \text{\textnormal{if }} x\in[a,a+m), \\
    y \quad & \text{\textnormal{if }} x\in[a+m,L].
\end{cases}
\end{equation}

\noindent \textbf{Perturbation Function-2 (\texttt{A2B1}).}
Given $C_p, a \in \mathcal{L}$ determined by $C_p$, $k\in [0,1)$, $m \geq 0, L>0, y>0$, $\forall x \in \mathcal{L}$, perturbation function $P(x)$ is as follows
\begin{equation}\footnotesize
    P(x) = 
\begin{cases} 
    y \quad & \text{\textnormal{if }} x\in[-L,a), \\
    y+k \mathrm{sin} \frac{\pi}{m}(x-a) \quad & \text{\textnormal{if }} x\in[a,a+m), \\
    y \quad & \text{\textnormal{if }} x\in[a+m,L].
\end{cases}
\end{equation}

\noindent \textbf{Perturbation Function-3 (\texttt{A3B1}).}
Given $C_p, a \in \mathcal{L}$ determined by $C_p$, $k\in [0,1)$, $m \geq 0, L>0, y>0$, $\forall x \in \mathcal{L}$, perturbation function $P(x)$ is as follows
\begin{equation}\footnotesize
    P(x) = 
\begin{cases} 
    y \quad & \text{\textnormal{if }} x\in[-L,a), \\
    y + \frac{2k}{m}x - \frac{2ak}{m} \quad & \text{\textnormal{if }} x\in[a,a+\frac{m}{2}), \\
    y - \frac{2k}{m}x + \frac{2(a+m)k}{m} \quad & \text{\textnormal{if }} x\in[a+\frac{m}{2},a+m), \\
    y \quad & \text{\textnormal{if }} x\in[a+m,L].
\end{cases}
\end{equation}

\noindent \textbf{Perturbation Function-4 (\texttt{A1B2}).}
Given $C_p, a \in \mathcal{L}$ determined by $C_p$, $k\in [0,1)$, $m \geq 0, L>0, y>0, t=\frac{y+k}{e^{\epsilon}}$, $\forall x \in \mathcal{L}$, perturbation function $P(x)$ is as follows
\begin{equation}\footnotesize
    P(x) = 
\begin{cases} 
    -\frac{y-t}{L^4}x^4 + y \quad & \text{\textnormal{if }} x\in[-L,a), \\
    -\frac{y-t}{L^4}x^4 + y + k \quad & \text{\textnormal{if }} x\in[a,a+m), \\
    -\frac{y-t}{L^4}x^4 + y \quad & \text{\textnormal{if }} x\in[a+m,L].
\end{cases}
\end{equation}

\noindent \textbf{Perturbation Function-5 (\texttt{A2B2}).}
Given $C_p, a \in \mathcal{L}$ determined by $C_p$, $k\in [0,1)$, $m \geq 0, L>0, y>0, t=\frac{y+k}{e^{\epsilon}}$, $\forall x \in \mathcal{L}$, perturbation function $P(x)$ is as follows
\begin{equation}\footnotesize
    P(x) = 
\begin{cases} 
    -\frac{y-t}{L^4}x^4 + y \quad & \text{\textnormal{if }} x\in[-L,a), \\
    -\frac{y-t}{L^4}x^4 + y + k \mathrm{sin} \frac{\pi}{m}(x-a) \quad & \text{\textnormal{if }} x\in[a,a+m), \\
    -\frac{y-t}{L^4}x^4 + y \quad & \text{\textnormal{if }} x\in[a+m,L].
\end{cases}
\end{equation}

\noindent \textbf{Perturbation Function-6 (\texttt{A3B2}).}
Given $C_p, a \in \mathcal{L}$ determined by $C_p$, $k\in [0,1)$, $m \geq 0, L>0, y>0, t=\frac{y+k}{e^{\epsilon}}$, $\forall x \in \mathcal{L}$, perturbation function $P(x)$ is as follows
\begin{equation}\footnotesize
    P(x) = 
\begin{cases} 
    -\frac{y-t}{L^4}x^4 + y \quad & \text{\textnormal{if }} x\in[-L,a), \\
    -\frac{y-t}{L^4}x^4 + y + \frac{2k}{m}x - \frac{2ak}{m} \quad & \text{\textnormal{if }} x\in[a,a+\frac{m}{2}), \\
    -\frac{y-t}{L^4}x^4 + y - \frac{2k}{m}x + \frac{2(a+m)k}{m} \quad & \text{\textnormal{if }} x\in[a+\frac{m}{2},a+m), \\
    -\frac{y-t}{L^4}x^4 + y \quad & \text{\textnormal{if }} x\in[a+m,L].
\end{cases}
\end{equation}

These six perturbation functions proposed in this work are merely illustrative examples. In practical applications, users can design perturbation functions following Definition~\ref{def-perturbation} to satisfy the criteria of boundedness and unbiasedness, and to ensure the privacy guarantee, as indicated in Theorems~\ref{unbiased_theorem} and \ref{DP_theorem}. To put it differently, these two theorems provide users with considerable flexibility in formulating the perturbation functions they utilize, enabling them to tailor their solutions to meet the specific demands of real-world scenarios.  This characteristic represents a central advantage of the composite DP mechanism.

\subsection{Privacy Verification of the Proposed Perturbation Functions}\label{sec_proof_DP}
In this section, we verify that the proposed perturbation functions satisfy $\epsilon$-differential privacy. 

\begin{corollary}
The proposed perturbation functions of the Composite DP mechanism preserve $\epsilon$-DP.
\end{corollary}
\begin{proof}
Given that $P$ and $P^\prime$ are the perturbation functions of $\mathcal{M}_c(D)$ and $\mathcal{M}_c(D^\prime)$, $H$ and $H'$ are activation functions of $P$ and $P'$, and $G$ is the base function. Let $v, v' \in \mathcal{N}$ denote the perturbed results of $P$ and $P'$. 
When using the base function \texttt{B1} with any of the three activation functions (\ie \texttt{A1} to \texttt{A3}), for any subset of outputs $S \subseteq \mathcal{N}$, we have
\begin{subequations}
\small
\begin{align}
    & \frac{\Pr(\mathcal{M}_c(D) \in S)}{\Pr(\mathcal{M}_c(D') \in S)} \notag\\
    & = \frac{P(v|\gamma^{-1}(v) \in S)}{P'(v'|\gamma^{-1}(v') \in S)} \notag\\  
    & = \frac{H(v|\gamma^{-1}(v)\in S)+G(v|\gamma^{-1}(v) \in S )}{H'(v'|\gamma^{-1}(v') \in S)+G(v'|\gamma^{-1}(v') \in S)} \notag \\
    & = \frac{\int_{u \in S} H(v|\gamma^{-1}(v) = u )du+y\cdot \int_{u \in S}1 d u}{\int_{u \in S} H'(v'|\gamma^{-1}(v') =u)du+y\cdot \int_{u \in S}1 d u}. \notag
\end{align}
\end{subequations}
\normalsize
Since 
\small
\begin{align*}
    H(v|\gamma^{-1}(v)=u)+y &\leq k+y \\
    & = \displaystyle\Sup_{x\in \mathcal{L}}P(x), 
\end{align*}
\normalsize
and
\small
\begin{subequations}
\begin{align*}
    H'(v'|\gamma^{-1}(v')=u)+y &\geq y \\
    & = \displaystyle\Inf_{x^\prime\in \mathcal{L}}~P'(x').
\end{align*}
\end{subequations}
\normalsize
According to Definition~\ref{def-perturbation}, we have
\small
\begin{subequations}
\begin{align*}
\frac{P(v|\gamma^{-1}(v)\in S)}{P'(v'|\gamma^{-1}(v')\in S)} & \leq \frac{(y+k)\cdot \int_{u \in S}1 d u}{y\cdot \int_{u \in S}1 d u} \\
& = \frac{\displaystyle\Sup_{x\in \mathcal{L}}~P(x)}
    {\displaystyle\Inf_{x^\prime\in \mathcal{L}}~P'(x')} 
    \leq e^{\epsilon}.
\end{align*}
\end{subequations}

\normalsize
When using the base function \texttt{B2} with any of the three activation functions (\ie \texttt{A1} to \texttt{A3}), for any subset of outputs $S \subseteq \mathcal{N}$, we have
\begin{subequations}
\small
\begin{align}
    & \frac{\Pr(\mathcal{M}_c(D) \in S)}{\Pr(\mathcal{M}_c(D') \in S)}\notag \\
    & = \frac{P(v|\gamma^{-1}(v) \in S)}{P'(v'|\gamma^{-1}(v') \in S)} \notag\\ 
    & = \frac{H(v|\gamma^{-1}(v) \in S)+G(v|\gamma^{-1}(v) \in S)}{H'(v'|\gamma^{-1}(v') \in S)+G(v'|\gamma^{-1}(v') \in S)} \notag \\
    & = \frac{\int_{u \in S}\left[H(v|\gamma^{-1}(v) =u)-\frac{y-t}{L^4}(v|\gamma^{-1}(v) =u)^4+y \right]du}{\int_{u \in S}\left[ H'(v'|\gamma^{-1}(v') =u)-\frac{y-t}{L^4}(v'|\gamma^{-1}(v') =u)^4+y\right]du}. \notag 
\end{align}
\end{subequations}
\normalsize
Since
\begin{subequations}
\small
\begin{align*}
    & H(v|\gamma^{-1}(v)=u)-\frac{y-t}{L^4}(v|\gamma^{-1}(v)=u)^4+y \\
    & \leq k+y \\
    & = \displaystyle\Sup_{x\in \mathcal{L}}P(x), 
\end{align*}
\end{subequations}
\normalsize
and
\begin{subequations}
\small
\begin{align*}
    & H'(v'|\gamma^{-1}(v')=u)-\frac{y-t}{L^4}(v'|\gamma^{-1}(v')=u)^4+y \\
    &\geq H'(v'|\gamma^{-1}(v')=u) + y \\
    & \geq y \\
    & = \displaystyle\Inf_{x^\prime\in \mathcal{L}}~P'(x'), 
\end{align*}
\end{subequations}
\normalsize
According to Definition~\ref{def-perturbation}, we have

\begin{subequations}
\small
\begin{align*}
\frac{P(v|\gamma^{-1}(v)\in S)}{P'(v'|\gamma^{-1}(v')\in S)} & \leq \frac{(y+k)\cdot \int_{u \in S}1 d u}{y\cdot \int_{u \in S}1 d u} \\
& = \frac{\displaystyle\Sup_{x\in \mathcal{L}}~P(x)}
    {\displaystyle\Inf_{x^\prime\in \mathcal{L}}~P'(x')} 
    \leq e^{\epsilon}.
\end{align*}
\end{subequations}

It is now easy to verify general cases that go beyond the proposed six perturbation functions. With any activation and base functions, the composition of which is based on Definition~\ref{def-perturbation}, the resulting Composite DP mechanism preserves $\epsilon$-DP.
\end{proof}

\section{Composite Differentially Private Mechanism in Practice}

In this section, we demonstrate the application and implementation of the proposed composite differentially private mechanism in real query scenarios.

\subsection{The Overall Workflow of Composite DP}
\begin{figure}[t]
  \centering
  \includegraphics[width=\linewidth]{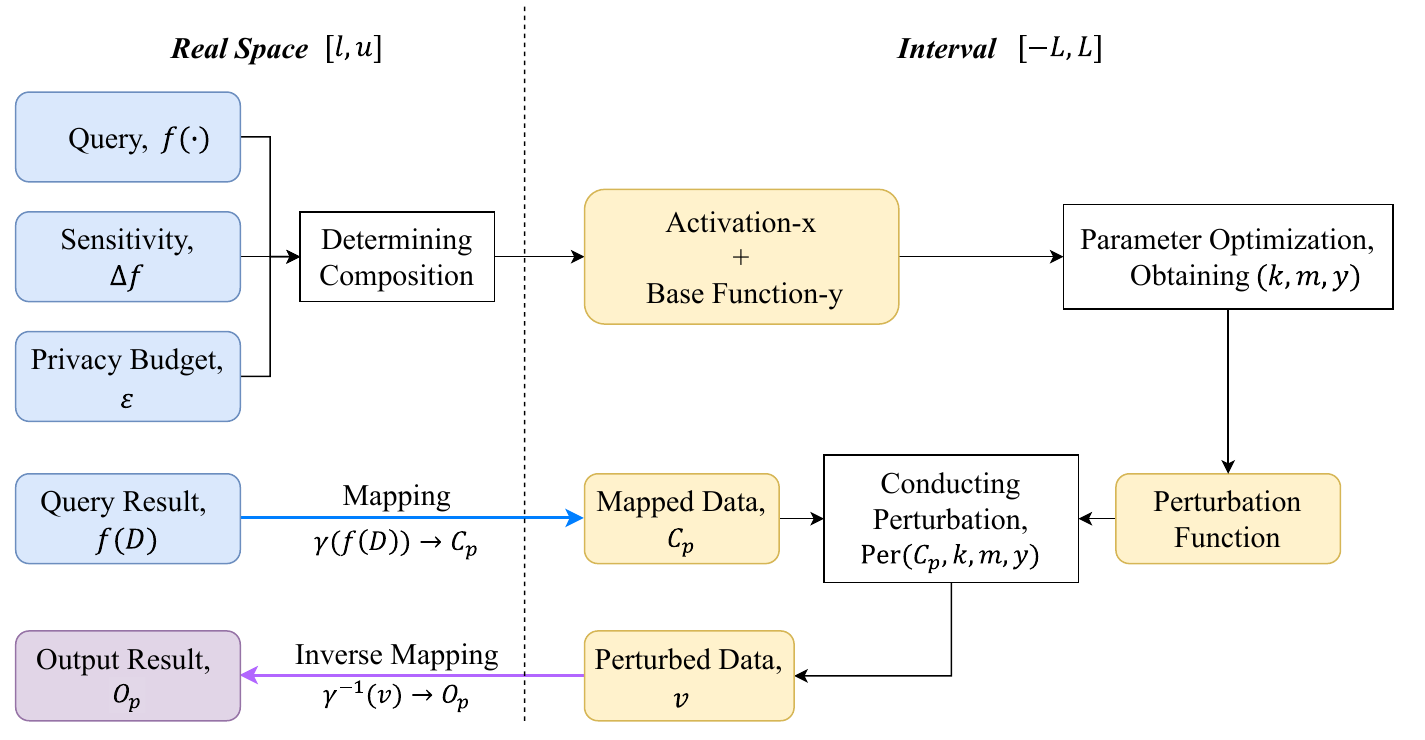}
  \caption{Workflow of the composite DP.}
  \label{Fig_workflow}
\end{figure}

In this section, we provide a comprehensive elucidation on its application to a real query scenario. The overall workflow of the composite differentially private mechanism is demonstrated in Figure~\ref{Fig_workflow}.

Given a query $f(\cdot)$, $\ell_1$-sensitivity value $\Delta f$, and privacy budget $\epsilon$, we first determine the composition by carefully choosing an activation function $H(x)$ and a base function~$G(x)$. Subsequently, we obtain the best parameter setup of $k, m, y$ according to the parameter optimization algorithm (detailed in the following Section~\ref{sec:opt}). Then we map the query result $f(D)$ from real space range~$\mathcal{N}$ into domain $\mathcal{L}$, and $C_p$ is the corresponding point to the $f(D)$. Afterward, we conduct the data perturbation given selected parameters $k, m, y$ and $C_p$. Following that, we inversely map the perturbed data $v$ back into real space $\mathcal{N}$. Finally, we publish the final result $O_{p}$. The mapping and inverse mapping functions are presented in Equations~\eqref{map_fun_l} and \eqref{map_inverse_fun_l}. 

\subsection{Parameter Optimization}
\label{sec:opt}

\begin{algorithm}[t]\footnotesize
\SetAlgoLined
\KwIn{steps $W$, stopping condition of iteration times $T$.}
\KwOut{Best parameter setup $k^{\mathrm{best}}, m^{\mathrm{best}}, y^{\mathrm{best}}$.}
$k,m,y \gets \randominit()$\;\label{alg_line_init}
$V_{\mathrm{min}} \gets \variance(k,m,y)$\;\label{alg_line_var_init}
$k^{\mathrm{best}},m^{\mathrm{best}},y^{\mathrm{best}} \gets k,m,y$\;
\For{$w \in W$}{\label{alg_line_loop_starts}
\While{$j \leq T$}{
$k,m,y \gets \modify(k,m,y,w)$\;\label{alg_line_modify}
$\mathrm{Var}\gets \variance(k,m,y)$\;
\If{$\mathrm{Var} < V_{\mathrm{min}}$}
{\label{alg_line_check_min}
$V_{\mathrm{min}} \gets \mathrm{Var}$\;
$k^{\mathrm{best}},m^{\mathrm{best}},y^{\mathrm{best}} \gets k,m,y$\;\label{alg_line_update_para}
}
$j\gets j+1$\;
}
}\label{alg_line_loop_ends}
\Return $k^{\mathrm{best}}, m^{\mathrm{best}}, y^{\mathrm{best}}$\;
\caption{Parameter Optimization (Enumeration)}
\label{alg_optimization}
\end{algorithm}

\begin{algorithm}[t]\footnotesize
\SetAlgoLined
\KwIn{step $w$, stopping condition of variance target $V_{\mathrm{t}}$.}
\KwOut{Best parameter set $k^{\mathrm{best}}, m^{\mathrm{best}}, y^{\mathrm{best}}$.}
$k,m,y \gets \randominit()$\;
$V_{\mathrm{min}} \gets \variance(k,m,y)$\;
$k^{\mathrm{best}},m^{\mathrm{best}},y^{\mathrm{best}} \gets k,m,y$\;
\While{$V_\mathrm{{min}} \geq V_{\mathrm{t}}$}{
$k,m,y \gets \modify(k,m,y,w)$\;
$\mathrm{Var} \gets \variance(k,m,y)$\;
\If{$\mathrm{Var} < V_{\mathrm{min}}$}{
$V_{\mathrm{min}} \gets \mathrm{Var}$\;
$k^{\mathrm{best}}, m^{\mathrm{best}}, y^{\mathrm{best}} \gets k,m,y$\;
}
}
\Return $k^{\mathrm{best}}, m^{\mathrm{best}}, y^{\mathrm{best}}$\;
\caption{Parameter Optimization (Search)}
\label{alg_optimization_2}
\end{algorithm}

\begin{figure}[t]
\centering
\includegraphics[width=0.48\textwidth]{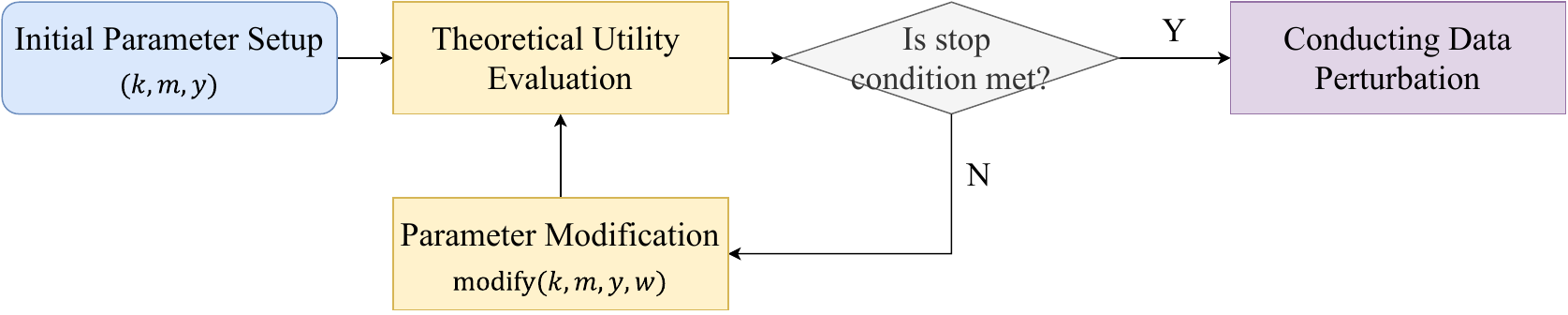}
\caption{Flow chart of parameter optimization.}
\label{fig_parameter_optimization}
\end{figure}

In this section, we further propose a general optimization algorithm to work out the best parameter setup for the proposed mechanism. Concretely, the final performance of the proposed mechanism is largely determined by the combined effects of the parameter setup $k, m$, and $y$ ($k$ and $m$ refer to the height and width of the activation function, respectively, and $y$ denotes the highest value of the base function). As none of them has a certain linear relationship with outcomes, it is practically difficult to directly derive the mathematical formula and identify the best solution. Hence, we utilize the iterative update approach to address this issue.

In practice, the iterative update is not an entirely new approach to parameter optimization and utility enhancement. Some past research utilizes the iterative approach to modify the parameters and generate less noise in DP mechanisms. For example, Ye~\etal~\cite{5,23} use iterative models to update $k$ and $v$ values for PrivKV and PrivKVM*, and Zhang~\etal~\cite{34} use an iterative approach to adjust the $\epsilon$ value. In addition, Zhu~\etal~\cite{35} use iteration to limit the noise by a presetting threshold. However, these solutions demand repeating the query of the database in every iteration procedure and adding the noise. It ignores the accumulative consumption of the privacy budget during the iteration. Once the allowable privacy budget runs out, the privacy is leaked. Therefore, we adopt the theoretical approaches of variance estimation to the iteration process for utility evaluation in a non-experimental way. The detail of the theoretical utility evaluation is given in Section~\ref{utility_eval}.

During the optimization process, as parameters $k, m$, and $y$ are independent of the $C_p$, we can adopt a hypothetical value for $C_p$ to compare variances across different parameter configurations (\eg $C_p = 0$, devoid of any influence from the actual input). As shown in Figure~\ref{fig_parameter_optimization} and Algorithm~\ref{alg_optimization}, a random parameter setup will be assigned to $(k, m, y)$ as initialization (line~\ref{alg_line_init}). A theoretical utility evaluation is further conducted to obtain the variance according to the functions determined by $(k, m, y)$ (line~\ref{alg_line_var_init}). Then, the parameter optimization process will start iteratively searching for a parameter setup that can obtain the minimum output variance until it reaches the largest iteration times $T$ (lines~\ref{alg_line_loop_starts} to~\ref{alg_line_loop_ends}). Particularly, the algorithm will modify the parameters according to step $w$ (line~\ref{alg_line_modify}). The modification function is determined by the iterative approaches; for example, if users enumerate all possible solutions, the modification function will orderly revise each parameter in a specific step size. If users utilize SGD, the modification function will randomly revise the parameter to lower the cost function, and then work out the theoretical variance. If a smaller variance is obtained, the corresponding parameters will be stored (lines~\ref{alg_line_check_min} to~\ref{alg_line_update_para}). When the stopping condition is satisfied, the algorithm will return the optimal parameter setup ($k^{\mathrm{best}}, m^{\mathrm{best}}, y^{\mathrm{best}}$) that (relatively) minimizes the output variance.

We note that there could be various optimization procedures implemented in practice, according to different stopping conditions defined by users. In Algorithm~\ref{alg_optimization}, we provide an enumeration optimization method as an example, which utilizes a set of steps (\eg 0.1, 0.01, and 0.001) and a threshold of maximum iterations as the stopping condition for each step, aiming to traverse the parameter space as much fine-grained as possible. In Algorithm~\ref{alg_optimization_2}, we provide a simpler optimization method that uses a threshold value of variance $V_t$ as the stopping condition, which employs only one step value to search one available parameter setup, \ie the optimization will stop when a certain variance is obtained. We argue that the proposed parameter optimization algorithm is a general approach. If users utilize SGD or other regression methods, the stopping condition could be whether the cost function has reached convergence. Considering the instability of SGD~\cite{wu2023training,sun2022surprising}, we use the enumeration optimization method (\ie Algorithm~\ref{alg_optimization}) in our experiments since it stably achieves the best solution. 

\subsection{Utility Evaluation}\label{utility_eval}
The utility evaluation takes a significant role in constructing the DP mechanism. For any two perturbation mechanisms and their perturbation function $P(x), P'(x)$, which both satisfy $\epsilon$-DP. When applying the same privacy budget $\epsilon$, these two mechanisms can be regarded as providing the same guarantee of privacy protection. Under this situation, the less variance of perturbed results, the better utility is provided, because less variance indicates less perturbation and distortion for the raw result. 
We first conduct a theoretical analysis of our proposed mechanism and the results of the experimental evaluation will be given in Section~\ref{sec:eval}. 

\begin{definition}[\textbf{Theoretical Variance}]\label{def_var}
Given the $\ell_1$-sensitivity $\Delta f$ of the query $f$, $C_p,$ $\cpmax$, $\cpmin \in \mathcal{L}$, and $C = \frac{\cpmax - \cpmin}{\Delta f}$. For any arbitrary points $\forall x \in \mathcal{L}$ in the output space of the perturbation mechanism with $P$, the theoretical variance $\mathrm{Var}: \mathcal{L} \rightarrow \mathcal{N}$ is defined as 
\begin{equation}\label{equ:var}
\begin{aligned}
    \mathrm{Var}(x) & = \frac{E(x^2) - (E(x))^2}{C^2}\\
    & = \frac{\int_{-L}^{L} P(x)x^2 dx - C_p^2}{C^2},
\end{aligned}
\end{equation}
where $E(x) = C_p$.
\end{definition}

Using~\eqref{equ:var} in Definition~\ref{def_var}, we have derived the theoretical variances for all six proposed perturbation functions discussed in Section~\ref{sec:perfunc}. The corresponding results are summarized in Table~\ref{table1}, with the detailed derivation processes.

We also give two alternative forms, \ie  $H_{1}^{\mathrm{Rate}}$  and $H_{2}^{\mathrm{Rate}}$, for calculating the utility estimation.

\begin{definition}[\textbf{$\mathbf{H_{1}^{\mathrm{Rate}}}$}]
Given $S_1$ and $S_2$ to represent the integral values of the activation function $H(x)$ and the base function $G(x)$ in domain $\mathcal{L}$, the metric $H_{1}^{\mathrm{Rate}}$ is defined as
\begin{equation}
    H_{1}^{\mathrm{Rate}} = \frac{S_2}{S_1}.
\end{equation}
\end{definition}
$H_{1}^{\mathrm{Rate}}$ indicates the ratio of the area of the base function to the activation function.  
The smaller value of $H_{1}^{\mathrm{Rate}}$ indicates the lower the variance and the better utility. 

\begin{definition}[\textbf{$\mathbf{H_{2}^{\mathrm{Rate}}}$}]
Given the mapped data $C_p \in \mathcal{L}$, the $S_{\mathrm{A}}(q)$ denotes the area (integral value) from $C_p$ to any point $q \in \mathcal{L}$. Let $n_1$, $n_2$ be partition scales. The metric $H_{2}^{\mathrm{Rate}}$ is defined as
\begin{equation}
    H_{2}^{\mathrm{Rate}} =
    \begin{cases}
        \frac{S_{\mathrm{A}}(C_p-n_1(L+C_p))}{S_{\mathrm{A}}(C_p-n_2(L+C_p))},& \text{if }~ C_p \geq 0,\\ \\
        \frac{S_{\mathrm{A}}(C_p+n_1(L-C_p))}{S_{\mathrm{A}}(C_p+n_2(L-C_p))}, & \text{if }~ C_p < 0,
    \end{cases}
\end{equation}
where $S_{\mathrm{A}}(q) = \int_{q}^{C_p} P(x)dx$.

\end{definition}

Intuitively, an increase in the fluctuation magnitude of the base function tends to impose limitations on the building of the activation function and irregularities in generating the perturbation function, ultimately resulting in an amplification of variance. Therefore, if we partition the area from $C_p$ to $L$ into two parts and assign the value of the area ratio to $H_{2}^{\mathrm{Rate}}$, then for any two perturbation mechanisms $\mathcal{M}_{c}, \mathcal{M}'_{c}$ that have the same activation function, the lower value of $H_{2}^{\mathrm{Rate}}$ indicates the greater constant of the base function and the lower variance (\ie the better utility). Additionally, as we use interquartile range (IQR) to partition the area, $n_1$ and $n_2$ values can be $\frac{1}{4}$ and $\frac{3}{4}$, respectively.

\section{Experimental Evaluation}\label{sec:eval}
In this section, we first present the experiment setup, including evaluation metrics and datasets, and then report the experimental results and analysis.

\subsection{Evaluation Metrics}
We utilize the following metrics to evaluate the performance of composite DP mechanism, with respect to output distribution, query response accuracy, and output bias.

\noindent\textbf{Privacy budget ($\epsilon$).} The privacy budget $\epsilon$ denotes the intensity of privacy protection in the perturbation function, \ie a smaller value of $\epsilon$ provides more indistinguishable outputs, and therefore achieves stronger privacy protection. 

\noindent\textbf{Relative Errors ($\mathrm{RE}$) and Mean Square Errors ($\mathrm{MSE}$).} We employ $\mathrm{RE}$ and $\mathrm{MSE}$ to evaluate the utility performance of a DP mechanism. These metrics serve as indicators of the distance (\ie the error), between the raw results $r_i$ and the perturbed results $r'_i$. The $\mathrm{RE}$ and $\mathrm{MSE}$ are defined as  
\begin{equation}
\mathrm{RE} = \frac{1}{N} \sum_{i=1}^{N}|r_i - r'_i|,
\end{equation}
\begin{equation}
\mathrm{MSE} = \frac{1}{N} \sum_{i=1}^{N}(r_i - r'_i)^2.
\end{equation}
Therefore, with the same privacy budget $\epsilon$, the smaller $\mathrm{RE}$ and $\mathrm{MSE}$ values, the better performance of an adopted DP algorithm.

\noindent\textbf{Accuracy Loss ($\mathrm{AL}$).} 
The $\mathrm{AL}$ denotes the disturbance degree for each query as defined below 

\begin{equation}\label{eq_acc_loss}
\mathrm{AL}_i = \Big{|}\frac{r_i-r'_i}{r_i} \Big{|}.
\end{equation}

\noindent\textbf{$\mathbf{H_{1}^{\mathrm{Rate}}}$ and $\mathbf{H_{2}^{\mathrm{Rate}}}$.} As discussed in Section~\ref{utility_eval}, $H_{1}^{\mathrm{Rate}}$ and $H_{2}^{\mathrm{Rate}}$ are used to quickly evaluate the variance for the proposed mechanism. We further provide experimental evidence to demonstrate that $H_{1}^{\mathrm{Rate}}$ and $H_{2}^{\mathrm{Rate}}$ proposed in Section~\ref{utility_eval} are practicable in composite DP design.

\subsection{Experimental Datasets and Benchmarks}
We compare our approach with traditional DP mechanisms, including the Laplace and Gaussian DP mechanisms. Specifically, we consider three datasets in our experiments, which are public representative datasets containing sensitive data features.

\noindent\textbf{Intensive Care Units (ICUs)~\cite{raffa2022global}} is a dataset containing patient information, vitals, chronic illnesses, and comorbidities. It includes 91,714 patient samples with detailed records from different hospitals, aiming to provide knowledge about the chronic condition for informing clinical decisions to ultimately improve patients' survival outcomes.

\noindent\textbf{Running and Heart Rate Data (RAHRD)~\cite{29}} is a dataset about running and heart rate summary, stop, start, and lap event data. It indicates a person utilized the smart wearable device to record its body information indicators when the person is jogging on the way. It contains 5,624 samples recorded across four years from 2017 to 2021 of running activities.

\noindent\textbf{Diabetes Dataset~\cite{30}} is originally from the National Institute of Diabetes and Digestive and Kidney Diseases. This dataset is used to predict whether a patient has diabetes using diagnostic measurements. It has 769 samples, including the physiological index of Pregnancy, Glucose, Blood Pressure, Skin Thickness, and Insulin.

\subsection{Experimental Results and Analysis}
We conduct each experiment 1,000 times and report the average results. 

\begin{figure*}[t]
\centering
\includegraphics[width=1.0\linewidth]{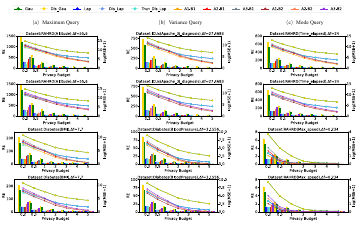}
\caption{$\mathrm{RE}$ and $\mathrm{MSE}$ values for the statistic queries.}
\label{fig_RE_MSE01}
\end{figure*}

\begin{figure*}[t]
\centering
\includegraphics[width=0.8\linewidth]{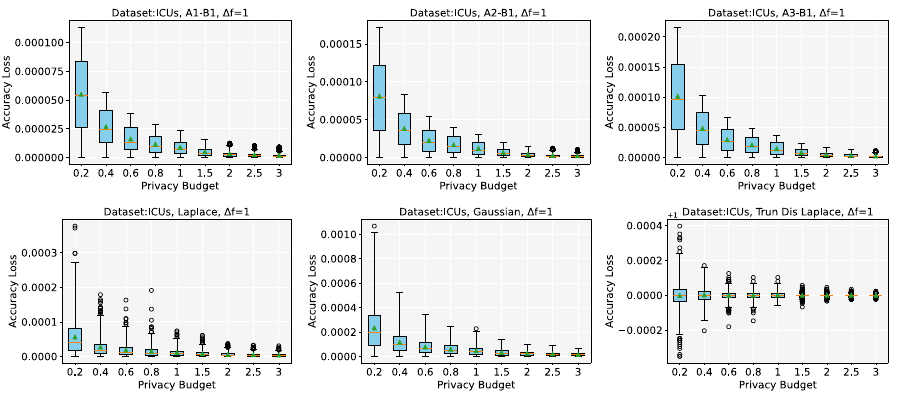}
\caption{Accuracy loss of the counting query (ICUs).}
\label{fig_acc_icus}
\end{figure*}

\begin{figure*}[t]
\centering
\includegraphics[width=0.8\textwidth]{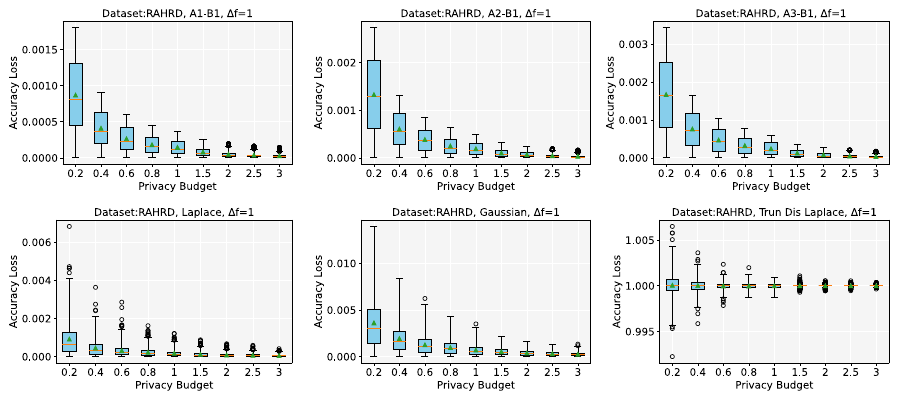}
\caption{Accuracy loss of the counting query (RAHRD).}
\label{fig_acc_rahrd}
\end{figure*}

\begin{figure*}[t]
\centering
\includegraphics[width=0.8\linewidth]{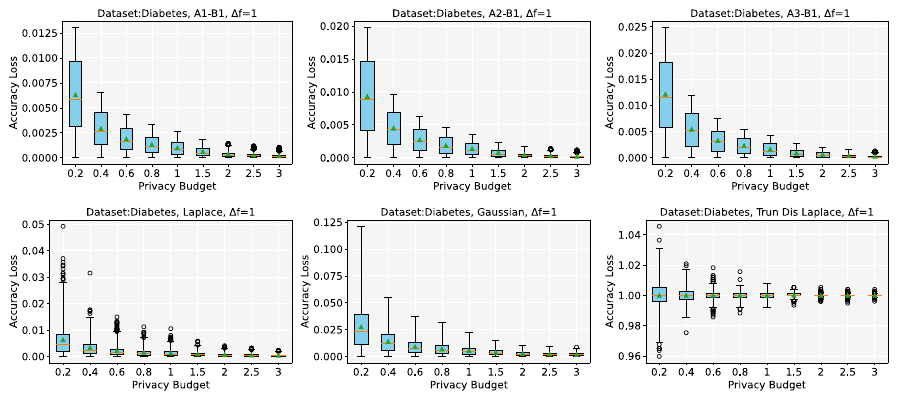}
\caption{Accuracy loss of the counting query (Diabetes).}
\label{fig_acc_diabetes}
\end{figure*}

\noindent\textbf{$\mathrm{RE}$ and $\mathrm{MSE}$ results.~} 
To evaluate the errors generated by DP mechanisms, we implement DP on queries of typical summary statistics, including Mode, Variance, Maximum, Minimum, and Mean. We evaluate our method within a broad range of $\epsilon$, \ie from 0.2 to 5, which are widely adopted in practical applications~\cite{zhang2022understanding,zhang2023robust,sander2023tan, bu2021fast,wei2022dpis,papernot2021tempered,nasr2021adversary,wang2022differential}. 

We present experimental results of Maximum, Variance and Mode queries in Figure \ref{fig_RE_MSE01}, and give more results of Minimum and Mean query in Appendix~\ref{appendix-B}. The experimental results indicate these statistics with the 6 perturbation functions presented in Section~\ref{sec:perfunc} (\ie \texttt{A1B1}, \texttt{A2B1}, \texttt{A3B1}, \texttt{A1B2}, \texttt{A2B2}, \texttt{A3B2}) and the comparison with the baselines (\ie the Laplace, the Discrete Laplace, the Truncated Discrete Laplace, the Gaussian and the Discrete Gaussian mechanisms~\cite{31,desfontaines2022differentially}). 
According to the experimental results, our proposed mechanism outperforms the Gaussian and the Discrete Gaussian mechanism. Among these combinations, 2 out of 6 ones outperform the Laplace, and the Discrete Laplace. The other 4 out of 6 combinations demonstrate weaker performance (compared to the Laplace, the Discrete Laplace and the Truncated Discrete Laplace) for small $\epsilon$ values, but exhibit stronger performance as $\epsilon$ values grow, across all evaluated summary statistics and datasets.

The combination \texttt{A1B1}  demonstrates the best performance among all combinations. On average, it surpasses the Gaussian mechanism by 85.68\% and 97.55\%, and surpasses the Discrete Gaussian mechanism by 87.07\% and 95.55\% in $\mathrm{RE}$ and $\mathrm{MSE}$, respectively. Additionally, it outperforms the Laplace mechanism (the Discrete Laplace mechanism has close performance to the Laplace mechanism) by 39.32\% in $\mathrm{RE}$ and the Truncated Discrete Laplace mechanism by 47.87\% in $\mathrm{MSE}$. Followed by \texttt{A1B2}, outperforms the Gaussian and the Discrete Gaussian mechanism by 85.64\% and 87.03\% in $\mathrm{RE}$, and outperform the Laplace and the Truncated Discrete Laplace mechanism by 65.23\% and 47.45\% in $\mathrm{MSE}$, respectively. The remaining four combinations, \texttt{A2B1}, \texttt{A2B2}, \texttt{A3B1}, \texttt{A3B2}, outshine the Gaussian mechanism by 80.62\%, 80.64\%, 76.53\%, 76.52\% in $\mathrm{RE}$, and 92.79\%, 92.55\%, 89.51\%, 89.48\% outshine the Discrete Gaussian mechanism in $\mathrm{MSE}$. However, they underperform the Laplace, the Discrete Laplace and the Truncated Discrete Laplace mechanisms when employing lower privacy budgets (\eg $\epsilon < 1$). 

Our method consistently exhibits enhanced performance across almost all settings. For example, for the Variance query, \texttt{A1B1} outperforms the Truncated Laplace mechanisms (the best-performing method among the baselines), reducing $\mathrm{MSE}$ by 90.34\% and 37.97\% when $\epsilon =5.0$ and $\epsilon =0.2$, respectively. For the Mode query, \texttt{A1B2} surpasses the Truncated Laplace mechanisms by 16.41\% and 35.89\% when $\epsilon =5.0$ and $\epsilon =0.2$, respectively.

\subsubsection{Effect of the activation function and base function} 
We find that different choices of the activation functions can significantly affect the performance. With the same base function, the activation function of a larger enclosed covered area  (\ie \texttt{A1} and \texttt{A2} in the proposed functions) will achieve significantly better performance as they yield a more concentrated probability output, thereby resulting in lower error rates. For example, the improvement gains of mode query for \texttt{A1B1} compared to \texttt{A3B1} in $\mathrm{RE}$ and $\mathrm{MSE}$ are 32.43\%, 55.91\%, and the improvement gains for \texttt{A2B2} compared to \texttt{A3B2} in $\mathrm{RE}$ and $\mathrm{MSE}$ are 14.63\%, 26.55\%, respectively.

In contrast, base functions have less impact on the performance as they only provide the basic probability for the output perturbation to guarantee privacy. For example, in the mode query, \texttt{A1B1} outperforms \texttt{A1B2} in $\mathrm{RE}$ by 0.60\%, and the improvement gains of \texttt{A2B1} in comparison to \texttt{A2B2} in $\mathrm{MSE}$ is 0.20\%. Among all two exemplar base functions, \texttt{B1} slightly outperforms the other one as it contributes less (\ie a relatively smaller proportion) to the final perturbation function with the same $k$ and $m$ values. 

\subsubsection{Effect of the sensitivity}
The $\ell_1$-sensitivity is determined by the nature of the datasets and the query functions (\ie the summary statistics in our experiments). 

We find that the sensitivity can affect the rate of descent of $\mathrm{RE}/\mathrm{MSE}$ with the increase of $\epsilon$ with respect to various combinations. The descent rate (unit: per $\epsilon$) regarding the increase of $\epsilon$ for the proposed mechanism \texttt{A1B1}, \texttt{A2B1}, \texttt{A3B1}, \texttt{A1B2}, \texttt{A2B2}, \texttt{A3B2} of $\mathrm{RE}/\mathrm{MSE}$. The maximum query in datasets RAHRD (attributes: Altitude) has a relatively higher sensitivity value of 55.6, and the $\mathrm{RE}$ descent rates among the proposed mechanisms \texttt{A1B1}, \texttt{A2B1}, \texttt{A3B1}, \texttt{A1B2}, \texttt{A2B2}, \texttt{A3B2} are 54.53, 83.78, 105.51, 54.1, 83.07, 105.68, respectively. While the mode query in datasets RAHRD (attributes: Max\_speed) has a relatively smaller sensitivity value of $0.234$, and the $\mathrm{RE}$ descent rates among these mechanisms are: 0.23, 0.35, 0.45, 0.23, 0.36, 0.44, respectively. 
The descent rate serves as an indicator of the statistical significance of the discrepancy between the output of a DP mechanism. When prioritizing a DP mechanism with a more significant output difference, \texttt{A3B1} and \texttt{A3B2} emerge as preferred selections; conversely, when prioritizing a mechanism with a smaller output difference, \texttt{A1B1} and \texttt{A1B2} are more favorable choices.

We also observe that a smaller sensitivity will result in a lower rate of descent. For example, for variance query on the Diabetes dataset (attributes: BloodPressure), which gives the sensitivity of 3.1556, the averaged performance of the proposed mechanisms achieves 22.60 and 703.57 of $\mathrm{RE}$ and $\mathrm{MSE}$ descending rates, while for mode query on the RAHRD Dataset (attributes: Time\_elapsed, which denotes the elapsed time since the initiation of running) which gives the sensitivity of 24, they achieve 172.77 and $4.12\times 10^{4}$ of $\mathrm{RE}$ and $\mathrm{MSE}$ descending rates. Consequently, when handling queries with high sensitivity, the selection of the mechanisms should carefully consider some specific circumstances, including the necessity of a more discernible output difference.

\noindent\textbf{Accuracy loss results.~} 
We then assess the accuracy loss ($\mathrm{AL}$) of Counting queries (\ie count the total number of records in the database).  We record $\mathrm{AL}$ for 1,000 queries and report the distributions of $\mathrm{AL}$ with box plots in Figures~\ref{fig_acc_icus} to \ref{fig_acc_diabetes}, where a tighter distribution refers to more controllable disturbance and thus better utility. 
Each group contains the Laplace, the Truncated Discrete Laplace and the Gaussian mechanisms and 3 perturbation functions including \texttt{A1B1}, \texttt{A2B1} and \texttt{A3B1}.   The employed privacy budget ranges from 0.2 to 3. The sensitivity of the Counting query is constant in this set of the experiments (\ie 1). The experimental results for the other three perturbation functions and the Discrete Laplace mechanism are demonstrated in Appendix \ref{appendix-C}.

In general, the performance of the averaged $\mathrm{AL}$ is in line with the results of $\mathrm{RE}$ and $\mathrm{MSE}$. \texttt{A1B1} yields the most favorable performance, follow by \texttt{A2B1} and \texttt{A3B1}. We find that all the 3 perturbation functions outperform the Laplace, the Truncated Discrete Laplace and Gaussian mechanisms in terms of the output distribution of $\mathrm{AL}$. \texttt{A1B1} performs best. It achieves the minimum output interval among the six mechanisms ($[0, 1.25\times 10^{-4}]$, $[0, 2\times 10^{-3}]$, $[0, 1.3\times 10^{-2}]$ in three datasets). \texttt{A3B1} has the maximum output interval ($[0, 2.5\times 10^{-4}]$, $[0, 4\times 10^{-3}]$, $[0, 2.5\times 10^{-2}]$ in three datasets). However, it still outperforms the  Laplace and Gaussian mechanisms. This outperformance can be attributed to the bounded feature of our approach, which enables better control over noise generation while maintaining the same privacy budget with less accuracy loss. It is worth noting that although the combinations with \texttt{A1} and \texttt{A2} (\eg \texttt{A1B1}, \texttt{A2B1}) have relatively small output ranges, these mechanisms may output outliers. For instance, the \texttt{A1B1} mechanism produces outliers on the ICUs, RAHRD, and Diabetes datasets, resulting in average accuracy losses of $1.44 \times 10^{-5}\%$, $2.30 \times 10^{-4}\%$, $1.65 \times 10^{-3}\%$, respectively.

The \texttt{A2B1} mechanism also generates outliers on each dataset, with the accuracy losses of $2.06 \times 10^{-5}\%$, $3.34 \times 10^{-4}\%$, $2.40 \times 10^{-3}\%$, respectively. This implies that the mechanisms employing \texttt{A1} and \texttt{A2} exhibit more instability compared to those employing \texttt{A3}. For users prioritizing output stability over a smaller variance, \texttt{A1B1} and \texttt{A2B1} are thus not recommended choices. Instead, \texttt{A3B1} should be considered, as they produce highly stable output distributions devoid of any outliers, and possess a more uniformly distributed range of $\mathupper$ and $\mathlower$ bounds. 

\noindent\textbf{$\mathbf{H_{1}^{\mathrm{Rate}}}$ and $\mathbf{H_{2}^{\mathrm{Rate}}}$ results.} Here, we investigate the variance of the proposed mechanism compared to the baselines, and evaluate the effectiveness of $H_{1}^{\mathrm{Rate}}$ and $H_{2}^{\mathrm{Rate}}$ as the indicators of utility.
As shown in Table~\ref{tab_H1_H2}, the performance ranking on variance is \texttt{A1B1}, \texttt{A1B2}, \texttt{A2B1}, \texttt{A2B2}, \texttt{A3B1}, \texttt{A3B2} with averaged variance of 11.448, 13.468, 25.746, 32.848, 41.107, 53.560, respectively. The results are in line with the performance in $\mathrm{RE}$, $\mathrm{MSE}$, and $\mathrm{AL}$, \ie a smaller variance results in better utility. 

\begin{table}[t]\small
\centering
\caption{$H_{1}^{\mathrm{Rate}}$ and $H_{2}^{\mathrm{Rate}}$ results.}
\label{tab_H1_H2}
\resizebox{0.8\linewidth}{!}{
\begin{tabular}{cccccccc}
\toprule
& $\epsilon$ & A1B1 & A2B1 & A3B1 & A1B2 & A2B2 & A3B2 \\
\midrule
\multirow{5}*{$H_{1}^{\mathrm{Rate}}$} & 0.2 & 9.382  &  15.308  &  18.263  &  9.571  &  15.942  &  18.263 \\
~ & 0.3 & 6.082  &  9.103  &  11.713  &  6.400  &  10.335  &  13.065 \\
~ & 0.4 & 4.447  &  6.711  &  8.456  &  5.536  &  8.106  &  8.782 \\
~ & 0.5 & 3.456  &  5.166  &  6.481  &  4.017  &  5.724  &  7.159 \\ 
~ & 1  & 1.466  &  2.093  &  2.588  &  1.752  &  2.628  &  2.743 \\
\midrule
\multirow{5}*{$H_{2}^{\mathrm{Rate}}$} & 0.2 & 0.333 & 0.333 & 0.333 & 0.334 & 0.334 & 0.334 \\
~ & 0.3 & 0.333 & 0.333 & 0.333 & 0.334 & 0.334 & 0.335 \\
~ & 0.4 & 0.333 & 0.333 & 0.333 & 0.335 & 0.335 & 0.334 \\
~ & 0.5 & 0.333 & 0.333 & 0.333 & 0.335 & 0.335 & 0.334 \\
~ & 1 & 0.333 & 0.333 & 0.333 & 0.335 & 0.336 & 0.335 \\   
\midrule
\multirow{5}*{Variance} & 0.2 & 31.714  &  69.653  &  112.554  &  33.986  &  85.726  &  141.532 \\ 
~ & 0.3 & 13.162  &  31.343  &  49.844  &  15.872  &  35.971  &  77.200 \\
~ & 0.4 & 7.218  &  16.325  &  25.106  &  11.117  &  27.320  &  28.311 \\
~ & 0.5 & 4.225  &  9.705  &  15.412  &  5.331  &  13.070  &  17.624 \\
~ & 1 & 0.921  &  1.702  &  2.620  &  1.036  &  2.152  &  3.133 \\
\bottomrule
\end{tabular}}
\end{table}

We also compare the $H_{1}^{\mathrm{Rate}}$ and $H_{2}^{\mathrm{Rate}}$ among the proposed mechanisms. The $H_{1}^{\mathrm{Rate}}$ refers to the proportion of the base function area to that of the activation function, which implies that a higher proportion of the base function leads to a more dispersed distribution of the perturbation function and consequently increases the variance. For example, the mechanism \texttt{A2B2} has higher $H_{1}^{\mathrm{Rate}}$ compared to \texttt{A1B2}, which leads to an average increase of 19.380 in variance. In regard to the $H_{2}^{\mathrm{Rate}}$, it represents the ratio of the area of the perturbation function near and far from the input. A lower $H_{2}^{\mathrm{Rate}}$ corresponds to a more constant base function and leads to lower variances. For example, the mechanism \texttt{A3B1} has lower $H_{2}^{\mathrm{Rate}}$ than \texttt{A3B2}, which results in an average decrease of 12.453 in variance. In conclusion, for mechanisms with the same base function, a higher $H_{1}^{\mathrm{Rate}}$ leads to lower utility performance, while for mechanisms with the same activation function, a lower $H_{2}^{\mathrm{Rate}}$ results in better utility performance. 

We also evaluate the bias rate of the proposed mechanism, and the results demonstrate the attainability of the unbiased performance. More details on the results of the bias rates are provided in Appendix~\ref{appendix-bias}. 

\section{Related Work}
Differential privacy (DP) is a  mathematical framework that provides a rigorous privacy assurance, effectively mitigating the risks of individual information disclosure~\cite{8,31,desfontaines2022differentially}. 
In real-world applications where output range constraints are prevalent, post-processing methods and truncated DP mechanisms are commonly employed to address such limitations. 

Post-processing mechanisms are applied on the outputs of the perturbed data, restricting their range to the feasible region~\cite{10,11,12,13,14}. Hay~\etal~\cite{10} employed post-processing to sort and round the non-integer or negative value to refine outcomes to boost the accuracy of a general class of histogram queries.
Based on this study, Qardaji~\etal~\cite{11} utilized m$\mathrm{MSE}$ to examine the factors that affect the accuracy of hierarchical approaches. They employed post-processing to diminish noise and enhance the accuracy of stratified methods through the application of constraint reasoning techniques. Cormode~\etal~\cite{12} conducted a specific post-processing method to improve the query accuracy of noisy counts. Balle and Wang~\cite{13} revisited the Gaussian mechanism and equipped it with a novel post-processing step based on adaptive estimation techniques. Lee~\etal~\cite{14} formulated the post-processing step as a constrained maximum likelihood estimation problem to further improve its output accuracy. 
Moreover, the post-processing operations have been utilized to handle released data in various applications, including mobility data for transportation~\cite{15}, power grid obfuscation~\cite{16}, smart grid pricing~\cite{17,xue2024rai4ioe}, and trajectory information~\cite{18}. 

Similarly, truncated DP mechanisms exhibit comparable effects, primarily built upon the Laplace or Gaussian mechanisms~\cite{37}. 
The truncated mechanism typically entails adjusting the probability density function of the original noise distribution to ensure that the output values remain strictly within the valid range. Geng \etal~\cite{36} utilized a truncated Laplace mechanism to reduce noise amplitude by increasing the decay rate of probability density while maintaining $\epsilon$-DP. Bun~\etal~\cite{38} proposed a truncated Concentrated DP (tCDP) mechanism with limited outputs and stronger privacy guarantees. 

However, both post-processing approaches and truncated DP are prone to bias issues as they cause the expected value of the perturbed data to deviate from the original input~\cite{6}. Previous studies introduced diverse definitions of bias problems according to different specific application scenarios. For instance, Tran~\etal~\cite{tran2021decision} defined the bias as the disparity impact of noise on two particular applications, including an allotment problem that distributes a finite set of resources among entities, and a decision rule that assesses if a particular entity is eligible for certain benefits. In contrast, our definition focuses on the fundamental cause of the bias arising from bounded DP mechanisms, ensuring the expected value of perturbed outputs aligns with the raw inputs, which can be applied to generic scenarios. Consequently, the unbiasedness of our proposed mechanism encompasses the one proposed in Tran~\etal~\cite{tran2021decision} and is also adaptable to a wider range of scenarios. Wang~\etal~\cite{wang2020locally} employed an LDP mechanism combined with post-processing techniques to prevent the occurrence of negative values or a sum of frequencies not equaling 1 due to noise. In their context, this concept refers to consistency. However, the unbiased perturbation is not guaranteed. 

To conquer the limitations of post-processing or truncated mechanisms, we propose a novel composite DP mechanism that enables users to conduct controllable perturbation with bounded and unbiased outputs. Our method offers a higher degree of freedom as it allows users to define their perturbation function for perturbation, leading to an infinite set of DP mechanisms that satisfy this framework. Given that different use cases may require various definitions of utility, it is difficult to determine the best solution for the proposed mechanism. However, by decreasing the proportion of the base function (as presented by $H_{1}^{\mathrm{Rate}}$ and $H_{2}^{\mathrm{Rate}}$), users can improve the results in terms of variance and errors produced by the mechanism. To the best of our knowledge, this is the first exploration of building DP on the practical composite mechanism, and our work bridges the knowledge gap on how to embed the bounded output feature into the perturbation function with respect to DP mechanisms.

We also note that our proposed composite DP mechanism differs from the composition theory~\cite{39,40,41,42} 
in the realm of DP research. Composition theory involves the sequential application of multiple DP mechanisms. This approach ensures that privacy is preserved even with degraded parameters of individual mechanisms. For example, Vadhan and Wang~\cite{32} and Lyu~\cite{33} recently explored the concurrent composition in interactive DP mechanisms, highlighting the inherent feature of interactivity in DP. 

\section{Conclusion}

The Laplace and Gaussian DP mechanisms are contemporarily popular techniques for protecting data privacy in various real-world applications, due to their provable privacy guarantees. However, traditional DP mechanisms generate unbounded outputs, which deviate from physical reality in many specific use cases. Previously, users typically resorted to post-processing or truncated mechanisms to constrain the output range, but it would introduce bias problems. To address these issues, we propose a novel composite DP mechanism that enables controllable perturbation and produces bounded and unbiased outputs, sidestepping the drawbacks of post-processing and truncated mechanisms. We conducted a comprehensive evaluation of our proposed mechanism through theoretical analysis and experimental validation. Additionally, we proposed a general iterative parameter optimization algorithm to reduce the noise and enhance the utility of the mechanism. Our experimental results demonstrate that the proposed mechanism performs consistently well across diverse applications while simultaneously maintaining the same level of privacy protection strength. 
As there are no specific constraints applied to the distributions of the base and the activation functions (which offer a high degree of flexibility), the proposed mechanism has a loose analysis and may not fully achieve optimal utility. Applying specific distributions to the base and the activation function would enable a tighter analysis and potentially further improve utility. We hope that the bounded and unbiased composite DP mechanism developed in this paper can lay the foundation for future privacy-preserving studies.

\section*{Acknowledgments}
This work was supported by Australian Research Council (ARC) DP190101893, DP240103068, and LP180100758.

\bibliographystyle{IEEEtran}
\bibliography{reference}

\balance

\begin{appendices}

\begin{figure*}[t]
\centering
\includegraphics[width=\linewidth]{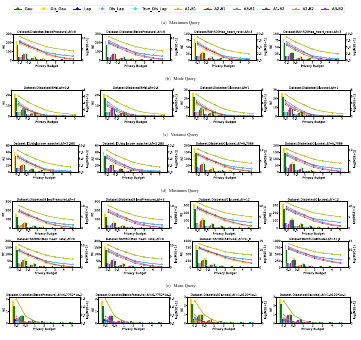}
\caption{Additional experimental results of $\mathrm{RE}$ and $\mathrm{MSE}$ for various statistic queries.}
\label{fig_RE_MSE02}
\end{figure*}

\begin{figure*}[t]
\centering
\includegraphics[width=\textwidth]{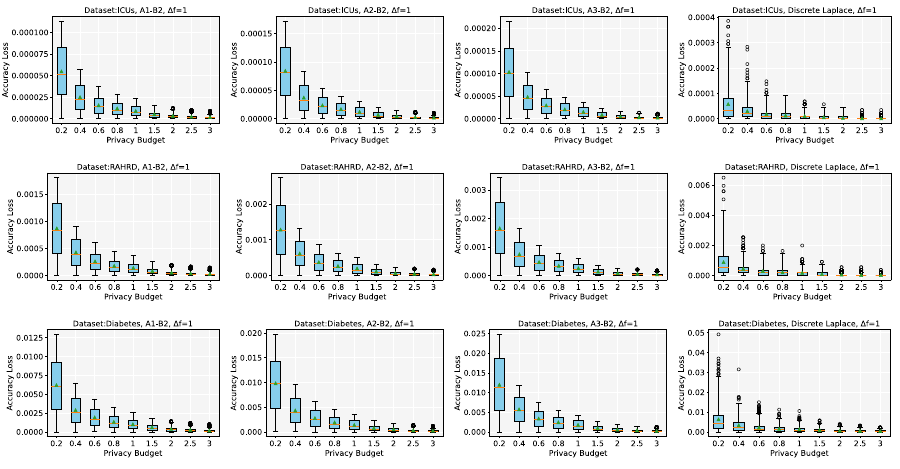}
\caption{Additional experimental results for accuracy loss.}
\label{fig_acc_appendix}
\end{figure*}

\section{Additional Experiments for $\mathrm{RE}$ and $\mathrm{MSE}$}\label{appendix-B}

This section demonstrates the additional experiments of $\mathrm{RE}$ and $\mathrm{MSE}$ values for the different statistic queries (\eg Maximum, Mode, Variance, Minimum, Mean queries). As shown in Figure~\ref{fig_RE_MSE02}, the experiments including six perturbation functions presented in Section~\ref{sec:perfunc} and the comparison with the baselines (\ie the Laplace, the Discrete Laplace, the Truncated Discrete Laplace, the Gaussian and the Discrete Gaussian mechanisms). The experimental trend is consistent with that described in the main text. The combination \texttt{A1B1} demonstrates the best performance among all combinations, followed by \texttt{A1B2}, \texttt{A2B1}, \texttt{A2B2}, \texttt{A3B1} and \texttt{A3B2}.

\section{Additional Experiments for $\mathrm{AL}$}\label{appendix-C}

This section demonstrates the experiments of accuracy loss estimation of the other three perturbation functions and the Discrete Laplace for the counting query. The employed privacy budget is from 0.2 to 3.0. 

As shown in Figure~\ref{fig_acc_appendix}, the experimental results in this section are similar to those in the main text. Due to the characteristic of unbounded outputs, the Discrete mechanism has more dispersed results, and abnormal values. The experimental result indicates that the Discrete Laplace mechanisms have greater accuracy loss than the proposed mechanisms.

\section{Additional Experiments for Bias Rate}\label{appendix-bias}
This section demonstrates an emperical performance in terms of unbiasedness. We calculate the bias rate according to Definition 2.4 (\ie the difference between the expected value of the perturbed data and the original data). For $\epsilon=0.2$, our mechanism's least favorable performance, \texttt{A3B1}, only exhibits a 0.156\% bias rate. At $\epsilon=5$, all proposed mechanisms demonstrate a 0.001\% bias rate, empirically verifying their unbiased nature.

\begin{table}[h]
\centering
\caption{Experimental results of bias rate}
\label{tab:bias_rate}
\resizebox{\linewidth}{!}{
\begin{tabular}{c|c|c|c|c|c}
\toprule
& \multicolumn{5}{c}{Bias Rate} \\
\hline
\diagbox{Method}{$\epsilon$ Value} & 0.2 & 0.5 & 1 & 2 & 5 \\
\midrule
A1B1 & 0.081\% & 0.030\% & 0.013\% & 0.005\% & 0.001\% \\
A2B1 & 0.123\% & 0.044\% & 0.018\% & 0.006\% & 0.001\% \\
A3B1 & 0.156\% & 0.055\% & 0.022\% & 0.007\% & 0.001\% \\
A1B2 & 0.081\% & 0.029\% & 0.013\% & 0.005\% & 0.001\% \\
A2B2 & 0.124\% & 0.044\% & 0.018\% & 0.006\% & 0.001\% \\
A3B2 & 0.155\% & 0.055\% & 0.022\% & 0.007\% & 0.001\% \\
\bottomrule
\end{tabular}}
\end{table}

\begin{figure*}[t]
\centering
\captionof{table}{Constraints on hyper-parameters of perturbation functions. 
}\label{table_appendix}
\includegraphics[width=\linewidth]{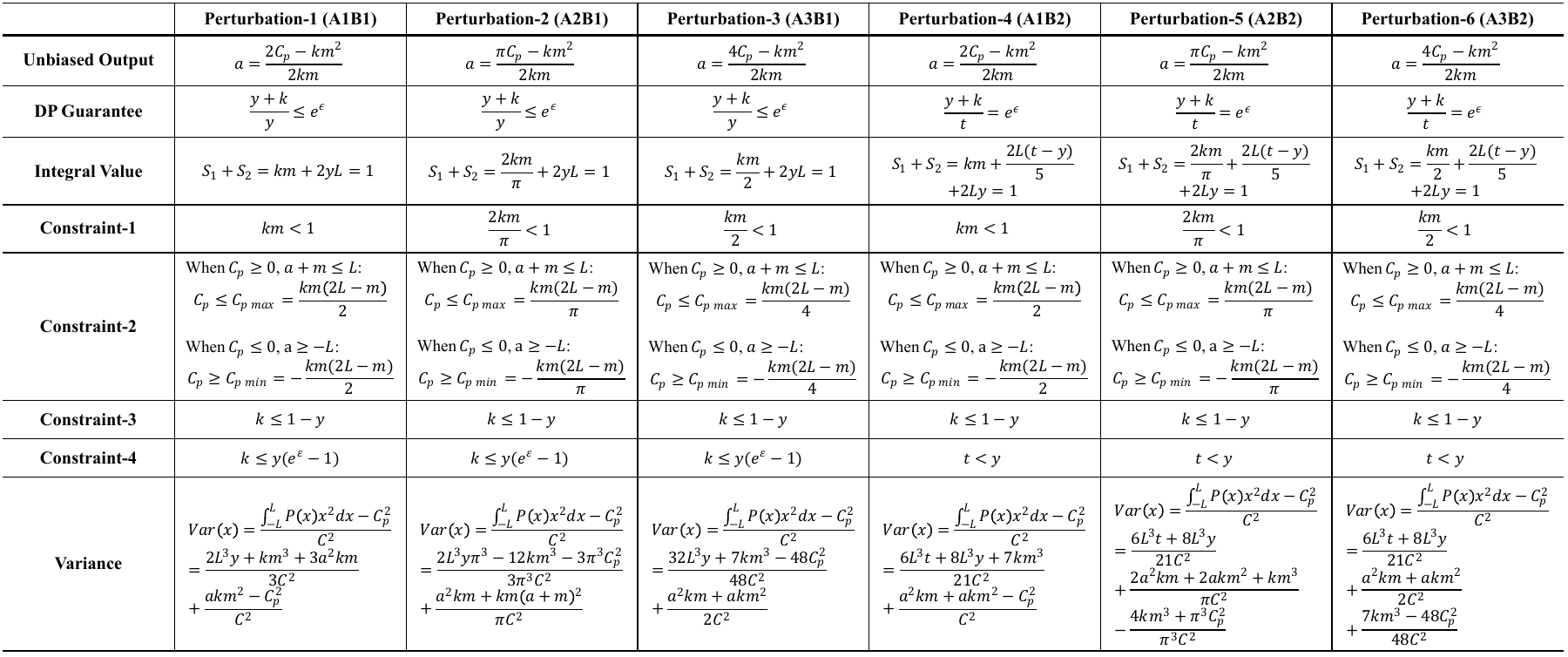}
\end{figure*}

\section{Perturbation Functions}\label{appendix-A}
In Table~\ref{table_appendix}, we present six perturbation functions and the corresponding embedded constraints on hyper-parameters (\ie $k$, $m$ and $y$) to satisfy the aforementioned unbiased and DP properties. 

\newpage
\newpage 
\clearpage
\section{Meta-Review}

The following meta-review was prepared by the program committee for the 2024
IEEE Symposium on Security and Privacy (S\&P) as part of the review process as
detailed in the call for papers.

\subsection{Summary}
This paper addresses the problem of unbiased private estimators for bounded domains by employing a base noise along with an additional activation noise designed to remove bias while still maintaining privacy guarantees. The effectiveness of the proposed mechanism is evaluated through both theoretical analysis and empirical experiments, demonstrating significant enhancements in terms of utility and privacy assurances in comparison to conventional Differential Privacy (DP) methods.

\subsection{Scientific Contributions}
\begin{itemize}
\item Creates a New Tool to Enable Future Science.
\item Addresses a Long-Known Issue.
\item Provides a Valuable Step Forward in an Established Field.
\end{itemize}

\subsection{Reasons for Acceptance}
\begin{enumerate}
\item The paper creates a new tool to enable future science. The proposed seminal composite DP mechanism produces bounded and unbiased outputs without the need for post-processing or truncation. This mechanism establishes a new building block for future privacy-preserving research. The authors have made the source code and the associated artifact publicly available, thereby expanding the DP toolkit and enabling future endeavors in the realm of privacy research.
\item The paper addresses a long-known Issue. The problem of bias introduced by truncation is a widely recognized issue. The approach developed in this paper introduces the first unbiased private estimator designed for bounded domains. It also affords a considerable level of flexibility, enabling users to define their perturbation functions and thereby construct a wide range of DP mechanisms. Additionally, the authors demonstrate a method for tuning the parameters in achieving variance minimization independently of the data prior to its application.
\item The paper provides a valuable step forward in an established field. The paper provides a new DP mechanism that outperforms other mechanisms according to the benchmarks presented in the paper. This evaluation involves assessing the variance of the composite probability density function and introducing alternative metrics that offer simplified computation compared to variance estimation. The comprehensive evaluation of this mechanism using three benchmark datasets highlights significant improvements when compared to conventional Laplace and Gaussian mechanisms.
\end{enumerate}

\subsection{Noteworthy Concerns} 
There are no specific constraints applied to the distributions of the base and the activation function. While this offers a high degree of flexibility, it leads to a loose analysis of the mechanism, indicating that it may not fully achieve optimal utility. Applying specific distributions to the base and the activation function would enable a tighter analysis and potentially further improve utility. However, the reviewers have reached the conclusion that it is of tremendous value and immense importance to publish this novel DP technique in its current form.

\end{appendices}

\end{document}